\documentclass[aps,onecolumn,a4paper,superscriptaddress,longbibliography,nofootinbib,notitlepage,groupedaddress]{revtex4-1}

\usepackage{lipsum}
\usepackage{setspace}

\usepackage{graphicx}
\usepackage{amsmath}
\usepackage{amssymb}
\usepackage{amsthm}
\usepackage[utf8]{inputenc}
\usepackage{comment}

\usepackage{color}

\newcommand{\norm}[1]{\left\lVert#1\right\rVert} % norm: double vertical bars

\usepackage[a4paper,total={6.5in,9in}]{geometry}
\usepackage[english]{babel}
\usepackage[T1]{fontenc}
\usepackage{mathrsfs}
\usepackage{braket}
\usepackage{physics}
\usepackage{enumerate}
\usepackage{graphicx}
\usepackage{mathtools}
\usepackage[export]{adjustbox}
\usepackage{microtype}

\theoremstyle{plain}

\usepackage[colorlinks = true,
            linkcolor = red,
            urlcolor  = blue,
            citecolor = blue,
            anchorcolor = red]{hyperref}

\newtheorem{theorem}{Theorem}
\newtheorem{corollary}[theorem]{Corollary}
\newtheorem{proposition}[theorem]{Proposition}
\newtheorem{lemma}[theorem]{Lemma}
 
\newtheorem{definition}[theorem]{Definition}  
\newtheorem{remark}[theorem]{Remark}  
\newtheorem*{remark*}{Remark}   %  \newtheorem*{remark}{Remark}  * is for numbering
\renewcommand\qedsymbol{$\blacksquare$}
\newenvironment{proof-of}[1][{\hspace{-\blank}}]{{\medskip\noindent\textit{Proof~{#1}.\ }}}{\hfill\qedsymbol}

\renewcommand{\Tr}{{\operatorname{Tr}\,}}

\newcommand{\id}{{\operatorname{id}}}

\newcommand{\1}{\openone}
\newcommand{\proj}[1]{|#1\rangle\!\langle #1|}
\newcommand{\cD}{{\mathcal{D}}}

\newcommand{\cT}{{\mathcal{T}}}

\newcommand{\KI}{{\text{KI}}}
\newcommand{\SW}{{\mathrm{SWAP}}}

\newcommand{\nc}{\newcommand}
\nc{\rnc}{\renewcommand}
\nc{\avg}[1]{\langle#1\rangle}
\nc{\Rank}{\operatorname{Rank}}
\nc{\smfrac}[2]{\mbox{$\frac{#1}{#2}$}}
\renewcommand{\tr}{\operatorname{Tr}}
\nc{\ox}{\otimes}
\nc{\dg}{\dagger}
\nc{\dn}{\downarrow}
\nc{\cA}{{\cal A}}
\nc{\cB}{{\cal B}}
\nc{\cC}{{\cal C}}
\nc{\cF}{{\cal F}}
\nc{\cG}{{\cal G}}
\nc{\cH}{{\cal H}}
\nc{\cI}{{\cal I}}
\nc{\cJ}{{\cal J}}
\nc{\cK}{{\cal K}}
\nc{\cL}{{\cal L}}
\nc{\cM}{{\cal M}}
\nc{\cN}{{\cal N}}
\nc{\cO}{{\cal O}}
\nc{\cP}{{\cal P}}
\nc{\cQ}{{\cal Q}}
\nc{\cR}{{\cal R}}
\nc{\cS}{{\cal S}}
\nc{\cX}{{\cal X}}
\nc{\cY}{{\cal Y}}
\nc{\cZ}{{\cal Z}}
\nc{\csupp}{{\operatorname{csupp}}}
\nc{\qsupp}{{\operatorname{qsupp}}}
\nc{\rar}{\rightarrow}
\nc{\lrar}{\longrightarrow}
\nc{\polylog}{{\operatorname{polylog}}}
\nc{\wt}{{\operatorname{wt}}}
\nc{\RR}{{{\mathbb R}}}
\nc{\CC}{{{\mathbb C}}}
\nc{\FF}{{{\mathbb F}}}
\nc{\NN}{{{\mathbb N}}}
\nc{\ZZ}{{{\mathbb Z}}}
\nc{\PP}{{{\mathbb P}}}
\nc{\QQ}{{{\mathbb Q}}}
\nc{\UU}{{{\mathbb U}}}
\nc{\EE}{{{\mathbb E}}}
\nc{\Hom}[2]{\mbox{Hom}(\CC^{#1},\CC^{#2})}
\nc{\rU}{\mbox{U}}

\nc{\ob}[1]{#1}

\nc{\SEP}{{\text{SEP}}}
\nc{\NS}{{\text{NS}}}
\nc{\LOCC}{{\text{LOCC}}}
\nc{\PPT}{{\text{PPT}}}
\nc{\EXT}{{\text{EXT}}}
\nc{\Sym}{{\operatorname{Sym}}}

\nc{\ERLO}{{E_{\text{r,LO}}}}
\nc{\ERLOCC}{{E_{\text{r,LOCC}}}}
\nc{\ERPPT}{{E_{\text{r,PPT}}}}
\nc{\ERLOCCinfty}{{E^{\infty}_{\text{r,LOCC}}}}
\nc{\Aram}{{\operatorname{\sf A}}}

\begin{document}

\title{\huge{Strong  Converse Bounds \protect\\  for Compression of Mixed States}}

%\title{\Large{On the Additivity of the Entanglement of Purification \protect\\ and Strong  Converse Bounds  for Compression of Mixed States}}

%Additivity of the entanglement of purification
%%%%%%%%%%%%%%%%%%%%%%%%%%%%%%%%%%%%%%%%%%%%%%%%%%%%%%%%%%%%%%%%%%%%%%%%%%%%%
\author{Zahra Baghali Khanian}
\email{zbkhanian@gmail.com }
\email{zkhanian@perimeterinstitute.ca}
\affiliation{Perimeter Institute for Theoretical Physics, Ontario, Canada, N2L 2Y5}
\affiliation{Institute for Quantum Computing, University of Waterloo, Ontario, Canada, N2L 3G1}
%\affiliation{Munich Center for Quantum Science and Technology \\
%\& Zentrum Mathematik, Technical University of Munich, Germany}
%{Technische Universit\"{a}t M\"{u}nchen, 85748 Garching, Germany }
%Email: zbkhanian@gmail.com
%TUM---
%\thanks{}
%\vspace{6cm}
\begin{abstract}
In this paper, we study strong converse properties for both visible and blind compression of mixed states.  
The optimal rate of a visible compression scheme is obtained in terms of the entanglement of purification,
whose additivity remains unknown so far.
For a variation of extendible states, we prove that the entanglement of purification is additive and apply
this  to obtain a ``pretty strong''  converse bound for the blind and visible compression of such states.
Namely, when the rate decreases below the optimal rate, the error exhibits a discontinuous jump from 0 to
at least $\frac{1}{3\sqrt{2}}$.

%We call a state $\rho^{AB}$  partially exchangeable  if $B$ can be decomposed into $A'B'$, via an isometry $U:B \to A'B'$, %as $\rho^{AA'B'}$,  and $A$ and $A'$ are exchangeable. Examples of partially exchangeable states are the output and the environment of Werner-Holevo channels with certain parameters. 
%, not isotropic states, Bell-diagonal states and so on. 

To deal with the visible case for general states, we define a new quantity $E_{\alpha,p}(A:R)_{\rho}$ for a bipartite state  $\rho^{AR}$  and $\alpha \in (0,1)\cup (1,\infty)$ as the $\alpha$-R\'enyi generalization of the entanglement of purification $E_{p}(A:R)_{\rho}$.  For $\alpha=1$, we define $E_{1,p}(A:R)_{\rho}:=E_{p}(A:R)_{\rho}$. We show that for any rate below  the regularization  $\lim_{\alpha \to 1^+}E_{\alpha,p}^{\infty}(A:R)_{\rho}:=\lim_{\alpha \to 1^+} \lim_{n \to \infty} \frac{E_{\alpha,p}(A^n:R^n)_{\rho^{\otimes n}}}{n}$ the fidelity for the visible compression% of ensembles of mixed states 
 exponentially converges to zero.
%%%%$\widetilde{E}_{p}^{\infty}(A:R)_{\rho}$, which is the regularized $\alpha$-R\'enyi entanglement of purification for $\alpha$ converging to 1 from right, 
%We conclude that if this regularized quantity is continuous with respect to $\alpha$, namely, if $\lim_{\alpha \to 1^+}E_{\alpha,p}^{\infty}(A:R)_{\rho}=E_{p}^{\infty}(A:R)_{\rho}$, then the strong converse holds for the visible compression of ensembles of mixed states.

Moreover, we  consider blind compression of a  general mixed-state source $\rho^{AR}$ shared between an encoder and an inaccessible reference system $R$.
%%%%%%%%, which is defined in \cite{ZBK_PhD,general_mixed_state_compression,ZK_mixed_state_ISIT_2020},
We obtain a strong converse bound for the compression of this source by assuming that the decoder is a super-unital channel. This immediately implies a strong converse for the blind compression of ensembles of mixed states, by assuming a super-unital decoder, as this is a special case of the general mixed-state source $\rho^{AR}$ where the reference system $R$ has a classical structure.
% %%%%%%%%%%considered in \cite{KI2001,KI2002}. 

%
%If one could show that $E_{\alpha,p}^{\infty}(A:R)_{\rho}$ is continuous from right  at $\alpha=1$, i.e. $\widetilde{E}_{p}^{\infty}(A:R)_{\rho}=E_{p}^{\infty}(A:R)_{\rho}$, then this would imply that strong converse holds for the visible compression of mixed states.
\end{abstract}
%\date{24 March 2025}
\maketitle
\section{Introduction and compression model}\label{sec: intro}

Quantum source compression was pioneered by Schumacher in \cite{Schumacher1995} where he provided
two definitions for a quantum source: 1. a quantum system together with correlations with a
purifying reference system, and 2. an ensemble of pure states. For both models, he showed that
the optimal compression rate is equal to the von Neumann entropy of the quantum system.
The compression of an ensemble source is respectively called visible or blind compression if the identity of the state is known to the encoder or not.
For ensemble of pure states, both blind and visible compression models lead to the same optimal compression rate \cite{Schumacher1995,Jozsa1994_1}.
Horodecki generalized the ensemble model of Schumacher to mixed states and found a multi-letter optimal rate considering the visible scheme \cite{Horodecki2000,Horodecki1998}. Later, Hayashi showed that this optimal rate is equal to 
the regularized entanglement of the purification of the source \cite{visible_Hayashi}.
Contrary to ensemble of pure states,  the blind and visible compression of the ensemble of mixed states do not lead to the same compression rate. In the blind model, the optimal rate is characterized by a single-letter rate formula in terms of a decomposition of the ensemble, called the Koashi-Imoto decomposition \cite{KI2001,KI2002}.

Recently, a general mixed-state source is considered in \cite{ZBK_PhD,ZK_mixed_state_ISIT_2020,general_mixed_state_compression} where a quantum source is defined as a quantum system together with correlations with a general quantum reference system. It is shown that the models defined by Schumacher as well as the ensemble of mixed states (blind scheme) are special cases of
this general mixed-state model. 
The optimal compression rate of this source is characterized by a single-letter entropic quantity in terms of a decomposition of the source which is a generalization of the Koashi-Imoto decomposition \cite{Hayden2004}.

Proving the optimality of the compression rates requires establishing converse bounds.
%In information theory, proving the optimality of a compression rate is called a converse proof.
In  so-called weak converse proofs, the compression rate is lower bounded in the limit of 
the error converging to zero. 
%
%One possible approach is obtaining the so-called weak converse bound, namely, in the limit of the error converging to zero, a lower bound is established for the compression rate.
%
However, weak converses leave open whether there is a trade-off between the compression rate and the error, that is if a smaller compression rate can be achieved at the expense of having a larger error. A strong converse, on the other hand, states that such a trade-off does not hold in the asymptotic limit of many copies of the source, that is for any rate below the optimal rate the error converges
to 1.
For  ensembles of pure states assuming that the decoder is a unitary, a strong converse was proved in \cite{Jozsa1994_1}. The assumption of the unitary decoder was lifted in \cite{Winter1999} where a strong converse bound is obtained using trace inequalities. 
However, despite all the progress in establishing strong converse bounds in various information-theoretic tasks, which deal with processing ensembles of pure states or
quantum systems with purifying reference systems e.g. \cite{Leditzky2016,Winter1999},
it was an open problem so far whether the strong converse holds for the compression of mixed states.
%
%establishing strong converse bounds for compression of mixed states was an open problems so far.  
%
%In this paper, we show that indeed the strong converse holds for both blind and visible compression of mixed states. 

% 
In this paper, we introduce partially-exchangeable states in Definition~\ref{def: XYZ}, and prove that the entanglement of purification is additive for these states.
Then, we show that a ``pretty strong'' converse bound (in the sense introduced in \cite{AW_pretty}) holds for both blind and visible compression of these states.  Namely, we show that the compression rate is lower bounded by the entropy of the system  for any error
$\epsilon< \frac{1}{3\sqrt{2}}$. 
We also deal with the visible compression of general states by introducing 
new quantity $E_{\alpha,p}(A:R)_{\rho}$ in Definition~\ref{def: E_alpha_p} 
as   an $\alpha$-R\'enyi generalization of the entanglement of purification
and obtain a bound on the fidelity with respect to the regularization of this quantity i.e. $E_{\alpha,p}^{\infty}(A:R)_{\rho}$.
We conclude that continuity of $E_{\alpha,p}^{\infty}(A:R)_{\rho}$ at $\alpha=1$ in the limit of $\alpha \to 1^+$ implies  that the strong converse  holds for the visible compression of mixed states. It remains an open question if $\lim_{\alpha \to 1^+}E_{\alpha,p}^{\infty}(A:R)_{\rho}=E_{p}^{\infty}(A:R)_{\rho}$ holds true.
Finally, we consider the blind  case and obatin a strong converse for the compression of a general mixed-state source
as defined in \cite{ZBK_PhD,ZK_mixed_state_ISIT_2020,general_mixed_state_compression} by assuming that the decoder is a super-unital channel. 
This immediately implies a strong converse for the blind compression of the ensemble of mixed states, by assuming super-unital decoders,  since this ensemble source is a special case of the general mixed-state source where the reference system has a classical structure.
%The proof is relatively simple 
%
%For specific class of sources, we show that the strong converse holds for any general CPTP (completely positive trace preserving) decoder.  
%
%Furthermore, we note that the strong converse bound that we obtain for the visible compression, holds for the blind compression as well since for the former model  the bound is obtained assuming that the encoder has a side information system, i.e.assuming a more powerful encoder.

%
%To this end, we use a result obtained in \cite{Dam2002} which related the fidelity between two states to their  R\'enyi entropies. The main hurdle 

\medskip
The organization of the paper is as follows. We introduce some notation and convention at the end of this section. In Section~\ref{sec: compression model} we define the asymptotic compression task which unifies the visible and blind model.
Section~\ref{sec: XYZ} is dedicated to the compression of partially-exchangeable states for which the entanglement of purification is additive.
In  Section~\ref{section: visible} and Section~\ref{section: blind}, we  address respectively the   visible and blind schemes. Finally, we discuss our results in Section~\ref{sec: discussion}.

\medskip

\textbf{Notation and Definitions.}  
In this paper, quantum systems are associated with finite dimensional Hilbert spaces $A$, $R$, etc.,
whose dimensions are denoted by $|A|$, $|R|$, respectively. 
%Since it is clear from the context, we slightly abuse the notation and let $Q$ denote both a quantum system and a quantum rate. 
We denote the set of normalized and sub-normalized (trace less than or equal to 1)  quantum states on Hilbert space $A$ by $\cS(A)$ and $\cS_{\leq}(A)$, respectively. 
%
%and its inverse function $\exp$, unless otherwise stated, is also to basis $2$).
%
%The conditional entropy and the conditional mutual information, $S(A|B)_{\rho}$ and $I(A:B|C)_{\rho}$, respectively, are defined in the same way as their classical counterparts: 
%\begin{align*}
%  S(A|B)_{\rho}   &= S(AB)_\rho-S(B)_{\rho}, \text{ and} \\ 
%  I(A:B|C)_{\rho} &= S(A|C)_\rho-S(A|BC)_{\rho} \\
%                    &= S(AC)_\rho+S(BC)_\rho-S(ABC)_\rho-S(C)_\rho.
%\end{align*}
%
%

The fidelity between two states $\rho$ and $\xi$ is defined as 
\(
 F(\rho, \xi) = \|\sqrt{\rho}\sqrt{\xi}\|_1 
                = \Tr \sqrt{\rho^{\frac{1}{2}} \xi \rho^{\frac{1}{2}}},
\) 
where  $\|X\|_1 = \Tr|X| = \Tr\sqrt{X^\dagger X} $ is the Schatten 1-norm. The purified distance between two states $\rho$ and $\xi$
is defined as $P(\rho, \xi):=\sqrt{1-F^2(\rho, \xi)}$.
It relates to the trace distance in the following well-known way \cite{Fuchs1999}:
\begin{equation}\label{lemma: Fuchs}
  1-F(\rho,\xi) \leq \frac12\|\rho-\xi\|_1 \leq \sqrt{1-F(\rho,\xi)^2}.
\end{equation}

%We identify states with their density operators, and we use the notation $\phi= \ketbra{\phi}{\phi}$ as the density operator of the pure state vector $\ket{\phi}$. 
For a state $\rho$, the von Neumann entropy and the $\alpha$-R\'enyi entropy for $\alpha\in (0,1)\cup (1,\infty)$ are defined respectively as
\begin{align}
    S(\rho) &:= - \Tr\rho\log\rho \\
    S_{\alpha}(\rho)&:= \frac{1}{1-\alpha}\log \Tr \rho^{\alpha}. 
\end{align}
Throughout this paper, $\log$ denotes by default the binary logarithm. 
%\begin{definition}[\cite{E_p}]

The entanglement of purification for a bipartite state $\rho^{AB}$ is defined as 
\begin{align}
    E_{p}(A:B)_{\rho}:=\inf_{V:C \to A'B'} S(AA')_{\sigma}, \nonumber
\end{align}
where the infimum  is over all isometries $V:C \to A'B'$, and $\sigma^{AA'}=\Tr_{BB'} \left((\1_{AB} \otimes V)\proj{\psi}^{ABC}(\1_{AB} \otimes V)^{\dagger}\right)$
%and $\rho^{AB}=\Tr_R \ketbra{\psi}^{ABC}$ 
for a purification $\ket{\psi}^{ABC}$ of $\rho^{AB}$.
%\end{definition}
In this definition, there is dimension bound on systems $A'$ and $B'$, hence, the infimum  is  attainable  \cite{E_p}.   

For a bipartite state $\rho^{AB} \in \cS_{\leq}(AB)$ the min-entropy of $A$ conditioned on $B$ is defined as
\begin{align}
H_{\min}(A|B)_{\rho}:=\max_{\sigma_B \in \cS(B)} \max \{\lambda \in \mathbb{R} : \rho^{AB} \leq 2^{- \lambda} \1 \ox \sigma^{B} \}.
\end{align}
With a purification $\ket{\psi}^{ABC}$ of $\rho^{AB}$, we define the max-entropy 
\begin{align}
H_{\max}(A|B)_{\rho}:=-H_{\min}(A|C)_{\psi^{AC}},
\end{align}
with the reduced state $\psi^{AC}=\Tr_B ({\psi}^{ABC})$.

Let $\epsilon \geq 0$
and $\rho^{AB} \in \cS_{\leq}(AB)$.
The $\epsilon$-smooth min-entropy of $A$ conditioned on $B$ is defined as
\begin{align}
H^{\epsilon}_{\min}(A|B)_{\rho}:= \max_{\rho'  \approx_{\epsilon} \rho} H_{\min}(A|B)_{\rho'},
\end{align}
where $\rho' \approx_{\epsilon} \rho$ means $P(\rho',\rho) \leq \epsilon$ for $\rho' \in \cS_{\leq}(AB)$. 
Similarly,
\begin{align}\label{eq: smooth duality}
H^{\epsilon}_{\max}(A|B)_{\rho}&:= \min_{\rho'  \approx_{\epsilon} \rho} H_{\max} (A|B)_{\rho'} \nonumber\\
&=-H^{\epsilon}_{\min}(A|C)_{\psi},
\end{align}
with a purification $\ket{\psi}^{ABC}$ of $\rho^{AB}$.

\section{Compression Model}\label{sec: compression model}
We consider a general finite-dimensional source defined as a quantum system $A$ that is correlated with a reference system 
$R$ in an arbitrary way, described by the overall state $\rho^{AR}$  \cite{ZBK_PhD,ZK_mixed_state_ISIT_2020,general_mixed_state_compression}. In particular, the 
reference does not necessarily purify the source, nor is it assumed to be classical. 
%
%Let $\ket{\psi}^{ARP}$ be a purification of the source. 
Let ${\rho}^{ARC}$ be an extension of the source. 
If the encoder has access to system $A$ as well as system $C$, we call the compression scheme visible.
Otherwise, if the encoder has access only to system $A$, we call the compression scheme blind. 
An example of the visible model, is the following ensemble source 
\begin{align}
    \rho^{ACR}=\sum_x p(x) \rho_x^A \otimes \ketbra{x}^C \otimes \ketbra{x}^R, \nonumber
\end{align}
considered in \cite{Horodecki2000,Horodecki1998,visible_Hayashi}, where the encoder has access to system $C$.

%An ensemble source can be defined as a classical-quantum state where preserving the fidelity with the classical reference system $R$ is equivalent to preserving the fidelity by taking the average over the states of the ensemble, as shown in Eq.~(\ref{eq:fidelity criterion}). Therefore, we consider the following classical-quantum state as the source
%\begin{align}
%    \rho^{ACR}:=\sum_x p(x) \rho_x^A \otimes \ketbra{c_x}^C \otimes \ketbra{x}^R, \nonumber
%\end{align}
%
%where $\{\ket{x}^R\}_x$ is an orthonormal basis for system $R$. We let systems $A$, $C$ and $R$ denote the system to be compressed, the side information of the encoder and the reference system, respectively. The compression is called blind or visible if $c_x=0$ for all $x$ or $c_x=x$ for all $x$, respectively. Hence, depending on the side information system this model includes both blind and visible schemes. 
%
We will consider the information theoretic limit of
many copies of the source $\rho^{AR}$, i.e.~$\rho^{A^n  R^n} = {(\rho^{AR})}^{\otimes n}$.
% \left(\rho^{ACR}\right)^{\otimes n}=\sum_{x^n} p(x^n) \rho_{x^n}^{A^n} \otimes \ketbra{c_{x^n}}^{C^n} \otimes \ketbra{x^n}^{R^n}$.
%
%We assume that the encoder, Alice, and the decoder, Bob, have initially a maximally entangled state $\Phi_K^{A_0B_0}$ on registers $A_0$ and $B_0$ (both of dimension $K$).
%
The encoder, Alice, performs the encoding compression operation 
$\mathcal{C}:A^n C^n \longrightarrow M $ on the systems $A^n C^n$ which is a quantum channel,
i.e.~a completely positive and trace preserving (CPTP) map. 
Notice that as functions CPTP maps act on the operators (density matrices) over 
the respective input and output Hilbert spaces, but as there is no risk of confusion,
we will simply write the Hilbert spaces when denoting a CPTP map.
Alice's encoding operation produces the state $\sigma^{M R^n}$ %=\sum_{x^n}p(x^n) \sigma_{x^n}^{M} \otimes \ketbra{x^n}^{R^n}$ 
with $M$  as the compressed system of Alice.
%
%The dimension of the compressed system is without loss of generality not larger than  the dimension of the original source, i.e. $|M| \leq  \abs{A}^n$. 
%
%The compression rate is defined as $\frac{\log |M|}{n}$.
%
The system $M$ is then sent to Bob via a noiseless quantum channel, who performs
a decoding operation $\mathcal{D}:M \longrightarrow \hat{A}^n $ on the compressed system 
$M$. The output system $\hat{A}^n $ is the reconstruction of the system $A^n$.
Bob's decoding operation produces the state $\xi^{\hat{A}^n R^n}$%=\sum_{x^n}p(x^n) \xi_{x^n}^{\hat{A}^n} \otimes \ketbra{x^n}^{R^n}$ 
 with system $\hat{A}^n$  as the reconstruction of system $A^n$. 
We say  the encoding-decoding scheme has the \emph{error} $ \epsilon$ or \emph{fidelity} $\sqrt{1-\epsilon^2}$ if
 \begin{align}
  \label{eq:fidelity criterion}
    \text{error}&:=P\left( \rho^{A^n R^n },\xi^{\hat{A}^n R^n} \right) \leq \epsilon,\\
  \overline{F}&:=F\left( \rho^{A^n R^n },\xi^{\hat{A}^n R^n} \right) \geq \sqrt{1-\epsilon^2}.
    %\sum_{x^n}p(x^n)F\left( \rho_{x^n}^{A^n},\xi_{x^n}^{\hat{A}^n} \right)  \geq 1-\epsilon,  
\end{align}
%where the second equality is due to the definition of the fidelity. 
 %In section~\ref{section: blind} where the side information system $C$ is trivial, $\overline{F}_b$ denote the fidelity of the blind scheme. Similarly, in section~\ref{section: visible} where we assume $C=R$,  $\overline{F}_v$ denotes the fidelity of the visible scheme. 
%where $\xi^{\hat{A}^n R^n A'_0 B'_0}=\left((\mathcal{D}\circ\mathcal{C})\otimes \id_{R^n}\right) \rho^{A^n R^n } \otimes \Phi_K^{A_0B_0}$.
%
%For the above compression task, $\frac1n \log|M|$ is called the  compression rate.
%
The minimum dimension $|M|$ of $M$ and the minimum rate, such that there exists an encoding-decoding scheme with error $\epsilon$,  are denoted by $\cM(n,\epsilon)$ and $\cQ(n,\epsilon):=\log \cM(n,\epsilon)/n$, respectively. 
The optimal (asymptotic) compression rate is defined as $\lim_{\epsilon \to 0}\lim_{n \to \infty}\cQ(n,\epsilon)/n$.
We use subscripts $v$ or $b$ to refer to visible or blind cases, respectively.

%
%Moreover, a rate $Q$ is called an (asymptotically) achievable rate  if for all $n$ there is a sequence of encoders and decoders such that the fidelity converges to $1$ and the compression rate converges to $Q$. The optimal rate is the infimum of all achievable rates.   
%The rate region is the set of all achievable rate pairs, as a subset of $\mathbb{R}_{\geq 0}\times\mathbb{R}_{\geq 0}$. 

According to Stinespring's theorem \cite{Stinespring1955}, a CPTP map 
$\cT: A \longrightarrow \hat{A}$ can be dilated to an isometry $U: A \hookrightarrow \hat{A} E$
with $E$ as an environment system, called an isometric extension of a CPTP map, such that 
$\cT(\rho^A)=\Tr_E (U \rho^A U^{\dagger})$. 
Therefore, the encoding and decoding operations  can in general be viewed as 
isometries $U_{\cC} : A^n  \hookrightarrow M W$ and
$U_{\cD} : M \hookrightarrow \hat{A}^n V$, respectively, 
with the systems $W$ and $V$ as the environment systems
of Alice and Bob, respectively.

%We remind that the task defined in section can be reduced to previously studied sources by considering  specific structures on the reference system.  Namely, if system $R$ purifies system $A$, then our source reduces to the purified source of Schumacher.  On the other hand, if system $R$ is classical, then our source model reduces to Schumacher's ensemble model and mixed-state ensemble model by Koashi and Imoto for pure states and mixed states on system $A$, respectively. 

\section{Compression of partially-exchangeable states}\label{sec: XYZ}

In this section, we  introduce partially-exchangeable states, and prove that the entanglement of purification is additive for these states.
Then, we show that a ``pretty strong'' converse bound (in the sense of \cite{AW_pretty}) holds for both blind and visible compression of these states.  
%
%We remind that a state $\rho^{AC}$ is called extendible with respect to system $A$ if there exists a state  $\sigma^{AA'C}$ such that it is invariant under the permutation of $A$ and $A'$, and $\tr_{A'}(\sigma^{AA'C})=\rho^{AC}$ \cite{Werner_extendible}.

%Reinhard F. Werner. An application of Bell’s inequalities to a quantum state extension problem. Letters in Mathematical Physics, 17(4):359–363, May 1989.
\subsection{Additivity of the entanglement of purification for partially-exchangeable states}\label{sec: additivity}

\begin{definition}\label{def: XYZ}
We call $\rho^{AB}$ a partially-exchangeable state if $B$ can be decomposed, via an isometry $U:B \hookrightarrow A'B'$ to systems $A'B'$, and $\proj{\varphi}^{AA'B'C}=(\1_{A}\ox U\ox \1_C)\proj{\psi}^{ABC}(\1_{A}\ox U\ox \1_C)^{\dagger}$ is invariant under the permutation of $A$ and $A'$ for any purification $\ket{\psi}^{ABC}$ of the state.
Namely, $\proj{\varphi}^{AA'B'C}=( \SW_{AA'}\ox \1_{B'C})\proj{\varphi}^{AA'B'C} ( \SW_{AA'}\ox \1_{B'C})^{\dagger}$ holds, where $\SW_{AA'}=\sum_{i,j}\ketbra{i}{j}^{A} \ox \ketbra{j}{i}^{A'}$.
\end{definition}
This definition is related to the (pure) extendibility of the state ${\varphi}^{A'B'C}$ with respect to $A'$. We  remind that a state $\rho^{AC}$ is called extendible with respect to system $A$ if there exists a state  $\sigma^{AA'C}$ such that it is invariant under the permutation of $A$ and $A'$, and $\tr_{A'}(\sigma^{AA'C})=\rho^{AC}$ \cite{Werner_extendible}.
Pure extendibility of the state ${\varphi}^{A'B'C}$ with respect to $A'$ implies that ${\varphi}^{AA'}$ has 
support entirely within the symmetric subspace (bosonic extension) or the antisymmetric subspace (fermionic extension) 
\cite{Myhr2009,Myhr_PhD,Kaur2018}.  
For a state ${\varphi}^{AA'}$, with support entirely within the symmetric or the antisymmetric subspace,
it is shown in \cite{Christandl_EoP} that the entanglement of purification is additive.
We show that  partially-exchangeable states have additive entanglement of purification as well.

\begin{lemma}
Let $\rho^{AB}$ be a partially-exchangeable state with a purification $\ket{\psi}^{ABC}$,
where $B$ can be decomposed, via an isometry $U:B \hookrightarrow A'B'$, into systems $A'$ and $B'$,
and the state $\proj{\varphi}^{AA'B'C}$ is invariant under the permutation of $A$ and $A'$.
Then, for $\rho^{AB}$, the entanglement of purification is $E_p(A:B){\rho} = S(A)_{\rho}$.
Thus, for these states, the entanglement of purification is additive, as
$E_p(A^n:B^n){\rho^{\ox n}} = nS(A)$ for any natural number $n$.
\end{lemma}
\begin{proof}
Let $\ket{\psi}^{ABC}$ be a purification of a partially-exchangeable state $\rho^{AB}$.  %
Also, let the isomtery $V:C \hookrightarrow C_AC_B$ be such that 
$E_p(A:B)=S(AC_A)=S(BC_B)$.
By definition, $B$ can be decomposed into $A'B'$, and $A$ and $A'$ are exchangeable.  
We use this to obtain the following lower bound
\begin{align}
2E_p(A:B)&=S(AC_A)+S(A'B'C_B)  \nonumber\\
&=S(A)+S(C_A|A)+S(A')+S(B'C_B|A')  \nonumber\\
&=2S(A)+S(C_A|A)+S(B'C_B|A)  \nonumber\\
&\geq 2S(A)+S(C_AC_B B'|A)  \nonumber\\
&= 2S(A)+S(C_AC_B B'A)-S(A)  \nonumber\\
&= 2S(A), 
\end{align}
where in the third line we use $S(A)_{\rho}=S(A')_{\rho}$.
The inequality is due to strong sub-additivity of the entropy.
Since $E_p(A:B)\leq S(A)$, hence we conclude that $E_p(A:B)=S(A)$.
Moreover, a tensor power state ${\rho^{\ox n}}$ is a partially-exchangeable state as well.
Therefore,  $E_p(A^n:B^n)_{\rho^{\ox n}}=S(A^n)=nS(A)$.
\end{proof}

\subsection{Pretty Strong Converse for compression of partially-exchangeable states}\label{sec: pretty}
%Pretty strong converse holds for any $\epsilon < \frac{1}{3\sqrt{2}}$
%\begin{align}
%2S(M)&=2S(\hat{A}^nV)_{\xi} \nonumber\\
%&=2S(A^nV)_{\rho} +\epsilon terms  \nonumber\\
%&=S(A^nV)_{\rho}+S(R^nV')_{\rho}  \nonumber\\
%&=S(A^nV)_{\rho}+S({A'}^n{R'}^nV')_{\rho}  \nonumber\\
%&=S(A^n)_{\rho}+S(V|A^n)_{\rho}+S({A'}^n)_{\rho}+S({R'}^nV'|{A'}^n)_{\rho}  \nonumber\\
%&= 2S(A^n)_{\rho}+S(V|A^n)_{\rho}+S({R'}^nV'|{A}^n)_{\rho}  \nonumber\\
%&\geq 2S(A^n)_{\rho}+S({R'}^nVV'|{A}^n)_{\rho}  \nonumber\\
%&= 2S(A^n)_{\rho}. 
%\end{align}
In this section, we consider the asymptotic compression of $\rho^{AR}$ by assuming that it is 
a partially-exchangeable state, and via an isometry $R$ can be decomposed into $A'R'$, and the purified state $\ket{\varphi}^{AA'R'C}$ is invariant 
under the permutation of $A$ and $A'$ systems. 
The rate $S(A)_{\rho}$ can be achieved, in both blind and visible compression schemes, by applying Schumacher
compression on $A^n$.
Moreover, for both blind and visible compressions of these states, we prove that the compression rate  is lower bounded  by  $S(A)_{\rho}$ for any error $\epsilon < \frac{1}{3\sqrt{2}}$. Hence, $S(A)_{\rho}$ is the optimal compression rate. More precisely the following converse bound holds.
\begin{theorem}
For both blind and visible compression  of partially-exchangeable  states, 
%there exists a constant $\mu$ such that 
for any error $\epsilon <\frac{1}{3\sqrt{2}}$, every integer $n$ and $\delta,\eta,\eta'>0$ and constant $\theta$
\begin{align}
\cM(n,\epsilon) &\geq nS({A})_{\rho} -3\log\frac{2}{\eta^2}- \frac{1}{2}\log\frac{2}{\eta'^2}  -\frac{1}{4}\log\frac{1}{\cos^2(2\alpha)}-\frac{\theta}{2}\sqrt{n \log \frac{2}{\delta}} \nonumber\\
&\geq nS({A})_{\rho} -O(\sqrt{n}),
\end{align}
where $\alpha=\sin^{-1}(3\epsilon+6\delta+6\eta+\eta')$.
\end{theorem}
\begin{proof}
Assuming that the code has the error $\epsilon $ in terms  of the  purified distance as $P(\rho^{A^nR^n},\xi^{\hat{A}^nR^n}) \leq \epsilon$, by Uhlmann's theorem
there is always an extension  $\rho^{A^nR^nV}$ of the source  satisfying $P(\rho^{A^nR^nV},\xi^{\hat{A}^nR^nV}) \leq \epsilon$.
This extension is obtained by applying an isometry $U^{C^n\to VV'}$ on a purification of the state as $\ket{\rho}^{A^nR^nVV'}=(\1_{A^nR^n}\ox U^{C^n\to VV'})\ket{\psi}^{A^nR^nC^n}$. Using this we obtain
\begin{align}\label{eq: smooth 1}
\cM(n,\epsilon)&\geq H_{\max}(M)_{\sigma} \nonumber\\
&= H_{\max}(\hat{A}^nV)_{\xi} \nonumber\\
&\geq H_{\max}^{\epsilon}({A}^nV)_{\rho} \nonumber\\
&  \geq H_{\min}^{\delta}({A}^n)_{\rho}+H_{\max}^{\epsilon+2\delta+2\eta}(V|{A}^n)_{\rho}-3 \log\frac{2}{\eta^2},
\end{align}
where the second line follows because applying the decoding isometry that does not change the max-entropy.
The third line is due to the definition of the smooth max-entropy and Uhlmann's theorem, as explained above. The last line is due the chain rule of Lemma~\ref{lemma: smooth chain rules},
which holds for any $\epsilon,\delta \geq 0$ and $\eta>0$.
From the third line of the above equations and that the state $\ket{\rho}^{A^n{A'}^n{R'}^nVV'}$ is pure we obtain
\begin{align}\label{eq: smooth 2}
\cM(n,\epsilon) &\geq H_{\max}^{\epsilon}({A}^nV)_{\rho} \nonumber\\
 &= H_{\max}^{\epsilon}({A'}^n{R'}^nV')_{\rho} \nonumber\\
&  \geq H_{\min}^{\delta}({A'}^n)_{\rho}+H_{\max}^{\epsilon+2\delta+2\eta}({R'}^nV'|{A'}^n)_{\rho}-3 \log\frac{2}{\eta^2}  \nonumber\\
&  = H_{\min}^{\delta}({A}^n)_{\rho}+H_{\max}^{\epsilon+2\delta+2\eta}({R'}^nV'|{A}^n)_{\rho}-3 \log\frac{2}{\eta^2},
\end{align}
where the inequality is due to the chain rule of Lemma~\ref{lemma: smooth chain rules}. We add  Eq.~(\ref{eq: smooth 1}) and Eq.~(\ref{eq: smooth 2}) and apply the second chain rule of Lemma~\ref{lemma: smooth chain rules},
for any $\epsilon,\delta \geq 0$ and $\eta,\eta'>0$  to obtain
\begin{align}
2\cM(n,\epsilon) & \geq 2H_{\min}^{\delta}({A}^n)_{\rho}+H_{\max}^{\epsilon+2\delta+2\eta}(V|{A}^n)_{\rho}
+H_{\max}^{\epsilon+2\delta+2\eta}({R'}^nV'|{A}^n)_{\rho} -6 \log\frac{2}{\eta^2}  \nonumber\\
&\geq 2H_{\min}^{\delta}({A}^n)_{\rho}
+H_{\max}^{3\epsilon+6\delta+6\eta+\eta'}({R'}^nV'V|{A}^n)_{\rho} -6 \log\frac{2}{\eta^2}- \log\frac{2}{\eta'^2}. 
\end{align}
We prove $H_{\max}^{3\epsilon+6\delta+6\eta+\eta'}({R'}^nV'V|{A}^n)_{\rho}\gtrapprox 0$.
Let $\sin \alpha=3\epsilon+6\delta+6\eta+\eta'$ and apply Lemma~\ref{lemma: H_min < H_max+terms} and duality in Eq.~(\ref{eq: smooth duality})
\begin{align}
H_{\max}^{\sin \alpha}({R'}^nV'V|{A}^n)_{\rho} &\geq H_{\min}^{\sin \alpha}({R'}^nV'V|{A}^n)_{\rho} -\log\frac{1}{\cos^2(2\alpha)} \nonumber\\
&= -H_{\max}^{\sin \alpha}({R'}^nV'V|{A}^n)_{\rho} -\log\frac{1}{\cos^2(2\alpha)}
\end{align}
This implies 
\begin{align}
H_{\max}^{3\epsilon+6\delta+6\eta+\eta'}({R'}^nV'V|{A}^n)_{\rho} &\geq  -\frac{1}{2}\log\frac{1}{\cos^2(2\alpha)}.
\end{align}
which holds for any $\alpha < \frac{\pi}{4}$. Hence,  the bound holds for $3\epsilon+6\delta+6\eta+\eta' <\frac{1}{\sqrt{2}}$. We can choose
$\delta, \eta,\eta'$ arbitrary small numbers, so the lower bound can hold for  $\epsilon <\frac{1}{3\sqrt{2}}$.
By applying AEP of Theorem~\ref{thm: AEP} we conclude 
\begin{align}
2\log \cM(n,\epsilon)& \geq 2H_{\min}^{\delta}({A}^n)_{\rho}+H_{\max}^{3\epsilon+6\delta+6\eta+\eta'}({R'}^nV'V|{A}^n)_{\rho} -6 \log\frac{2}{\eta^2}  \nonumber\\
&\geq 2H_{\min}^{\delta}({A}^n)_{\rho} -6 \log\frac{2}{\eta^2}- \log\frac{2}{\eta'^2}  -\frac{1}{2}\log\frac{1}{\cos^2(2\alpha)} \nonumber\\
&\geq 2nS({A})_{\rho} -6 \log\frac{2}{\eta^2}- \log\frac{2}{\eta'^2}  -\frac{1}{2}\log\frac{1}{\cos^2(2\alpha)}-\theta_{RC}\sqrt{n \log \frac{2}{\delta}}. \nonumber
\end{align}
\end{proof}

\begin{remark}
Consider the compression of a partially-exchangeable state $\rho^{AR}$ (again $A$ is the system to be compressed), where $A$ can be decomposed into $A'R'$, and the purified state $\ket{\psi}^{A'R'RC}$ is invariant 
under the permutation of $R$ and $R'$ systems. We can similarly prove that the compression rate  is lower bounded  by  $S(R)_{\rho}$ for any error $\epsilon < \frac{1}{3\sqrt{2}}$.
However, this rate  can be achieved only in the visible case, where the encoder has access to the purifying reference system $C^n$%
\end{remark}

%\textcolor{red}{Remark: The reason for having $\frac{1}{3\sqrt{2}}$ instead of $\frac{1}{\sqrt{2}}$ is that I use the chain rule twice whereas in the pretty strong paper by Morgan and Winter chain rule is used once. Using chain rules mutiple times limits epsilon for which the inequality holds}

%%%%%%%%%%%%%%%%%%%%%%%%%%%%%%%%%%%%%%%%%%%%%%%%%%%%%%%%%%%%%%%%%%%%%%%%%%%%%%%%%%%%%
\section{Visible compression case}\label{section: visible}
%\textcolor{red}{consider more general source with C}
In this section, we consider the visible compression case. 
The encoding-decoding model is defined in Section~\ref{sec: compression model}, and we assume that
the states of the side information system satisfy $\ket{c_x}^C=\ket{x}^C$ for all $x$, that is the encoder has access to the identity of the states in the ensemble.
The optimal compression rate
for this source is equal to the regularized entanglement of purification $E_p^{\infty}(A:R)_{\rho}:=\lim_{n\to \infty} \frac{E_p(A^n:R^n)_{\rho^{\otimes n}}}{n}$ of the source \cite{visible_Hayashi}.
%where $E_p(\cdot,\cdot)$ is defined as follows:
%
%The entanglement of purification is defined as follows:
%
We define an $\alpha$-R\'enyi generalization of the entanglement of purification. 
We use this quantity to obtain a strong converse bound.
\begin{definition}\label{def: E_alpha_p}
For $\alpha \in(0,1)\cup (1,\infty)$, the $\alpha$-R\'enyi entanglement of purification of a bipartite state $\rho^{AR}$ is defined as
\begin{align}
    E_{\alpha,p}(A:R)_{\rho}:=\inf_{\cN:B\to E} S_{\alpha}(AE)_{\sigma}, \nonumber
\end{align}
where the infimum  is over all CPTP maps $\cN$  such that $\sigma^{AE}=(\id_A \otimes \cN)\rho^{AB}$
and $\rho^{AB}=\Tr_R \ketbra{\psi}^{ABR}$ for a purification $\ket{\psi}^{ABR}$ of $\rho^{AR}$. For $\alpha=1$, define $E_{1,p}(A:R)_{\rho}:=E_{p}(A:R)_{\rho}$.

The regularized $\alpha$-R\'enyi entanglement of purification is defined as:
\begin{align*}
    E_{\alpha,p}^{\infty}(A:R)_{\rho}:=\lim_{n\to \infty } \frac{E_{\alpha,p}(A^n:R^n)_{\rho^{\otimes n}}}{n},
\end{align*}
where the right hand side is evaluated with respect to the state $(\rho^{AR})^{\otimes n}$.
\end{definition}

%\begin{remark}
%Any decreasing sequence of positive numbers is convergent, hence,  $\{\frac{E_{\alpha,p}(A^n:R^n)_{\rho^{\otimes n}}}{n}\}_n$ is convergent. 
%\end{remark}

\begin{lemma}
The dimension of the system $E$ in Definition~\ref{def: E_alpha_p} is upper bounded by $|A|^2\cdot|R|^2$. 
\end{lemma}
\begin{proof}
For $\alpha>1$ we have
\begin{align}
   E_{\alpha,p}(A:R)_{\rho}&=\inf_{\cN:B\to E} S_{\alpha}(AE)_{\sigma}, \nonumber\\
   &=\frac{1}{\alpha-1} \inf_{\cN:B\to E} -\log \Tr (\sigma^{AE})^{\alpha} \nonumber\\
   &=-\frac{1}{\alpha-1}  \log  \sup_{\cN:B\to E}\Tr (\sigma^{AE})^{\alpha} \nonumber
\end{align}
the last equality is due to the fact that $-\log(\cdot)$ is a decreasing function.
The map $\rho \mapsto \Tr \rho^{\alpha}$ is convex~\cite{Lieb2013,Tomamichel_book}. Hence, the supremum  in the last line can be  attained by extremal CPTP maps
which are maps that cannot be expressed as convex combination of other CPTP maps. 
Moreover, extremal CPTP maps with input dimension $d$ have at most $d$ operators in Kraus 
representation \cite{Choi1975}. Therefore, the dimension of the output system $E$ is bounded by $|B|^2=|A|^2\cdot |R|^2$.

For  $\alpha\in (0,1)$, the R\'enyi entropy $S_{\alpha}(AE)_{\sigma}$ is a concave function of quantum states,  hence, the infimum  in the definition of $E_{\alpha,p}(A:R)_{\rho}$ is   obtained by extremal CPTP maps.
As discussed above, the dimension of the output system $E$ is bounded by $|B|^2=|A|^2\cdot |R|^2$.
\end{proof}

The dimension bound on system $E$ implies that infimum is attained in Definition~\ref{def: E_alpha_p} since the set of CPTP maps with bounded input and output dimension is compact. 
\begin{corollary}
The infimum in the definition of the $\alpha$-R\'enyi entanglement of purification is attainable.
\end{corollary}

In the following lemma, we state some useful properties of the $\alpha$-R\'enyi entanglement 
of purification which we apply in the subsequent statements. 
Further properties of  this quantity are discussed in \cite{Renyi_E_p}.

\begin{lemma}\label{lemma: E_alpha_p properties}
The $\alpha$-R\'enyi entanglement of purification $E_{\alpha,p}(A:R)_{\rho}$ and its regularization $E_{\alpha,p}^{\infty}(A:R)_{\rho}$ have the following properties:
\begin{enumerate}[(i)]
    \item They are both decreasing with respect to $\alpha$ for $\alpha > 0$.
    %\item It is convex with respect to $\alpha$.
    \item $E_{\alpha,p}(A:R)_{\rho}$ is continuous with respect to $\alpha$ at any $\alpha\geq 1$.
    %\item $E_{\alpha,p}^{\infty}(A:R)_{\rho}$ is decreasing with respect to $\alpha$ for $\alpha > 0$.
    \item $E_{\alpha,p}^{\infty}(A:R)_{\rho}$ is continuous with respect to $\alpha$ at any $\alpha> 1$. 
\end{enumerate}
\end{lemma}
\begin{proof}
(i) The R\'enyi entropy $S_{\alpha}(AE)_{\sigma}$ is decreasing with $\alpha$, hence, for $\alpha'>\alpha$ we obtain
\begin{align}
    E_{\alpha,p}(A:R)_{\rho}&=\min_{\cN:B\to E} S_{\alpha}(AE)_{\sigma}\nonumber \\
    &\geq S_{\alpha'}(AE)_{\sigma_0}\nonumber \\
    &\geq \min_{\cN:B\to E} S_{\alpha'}(AE)_{\sigma}\nonumber \\
    &= E_{\alpha',p}(A:R)_{\rho}, \nonumber 
\end{align}
where in the second line $\sigma_0^{AE}=(\id_A \otimes \cN_0)\rho^{AB}$ for the optimal map
$\cN_0:B \to E$.
This implies that $E_{\alpha,p}^{\infty}(A:R)_{\rho}$ is decreasing with $\alpha$ as well.

\noindent  (ii) For a bipartite state $\rho^{AR}$ define $Z_{\alpha}(A:R)_{\rho}:=\max_{\cN:B\to E} \log \Tr (\sigma^{AE})^{\alpha}$ where $\sigma^{AE}=(\id_A \otimes \cN) \rho^{AB}$ and  $\rho^{AB}=\Tr_R \ketbra{\psi}^{ABR}$ for purification $\ket{\psi}^{ABR}$ of $\rho^{AR}$. We obtain 
\begin{align}
    E_{\alpha,p}(A:B)_{\rho}&=\min_{\cN:R\to E} S_{\alpha}(AE)_{\sigma}\nonumber \\
    &=\min_{\cN:R\to E} \frac{\log \Tr (\sigma^{AE})^{\alpha}}{1-\alpha}\nonumber \\
    &=\frac{1}{1-\alpha}\max_{\cN:R\to E} \log \Tr (\sigma^{AE})^{\alpha} \nonumber \\
    &=\frac{1}{1-\alpha}Z_{\alpha}(A:R)_{\rho}, \nonumber 
\end{align}
where the third follows for $\alpha>1$. In the following we show that $Z_{\alpha}(A:R)_{\rho}$ is a continuous function of $\alpha$ at any $\alpha>1$. It follows  that $E_{\alpha,p}(A:B)_{\rho}$ is a continuous function of $\alpha$ for $\alpha>1$ because $\frac{1}{1-\alpha}$ is continuous. To this end, we first show that $\alpha \mapsto Z_{\alpha}(A:R)_{\rho}$ is convex.
Let $\alpha=p \alpha_1+(1-p)\alpha_2$ for $p \in[0,1]$. We obtain
\begin{align}
    Z_{\alpha}(A:R)_{\rho}&= \log \Tr (\sigma^{AE})^{\alpha} \nonumber\\
    &= \log \Tr (\sigma^{AE})^{p\alpha_1+(1-p)\alpha_2} \nonumber\\
    &\leq p\log \Tr (\sigma^{AE})^{\alpha_1}+(1-p)\log \Tr (\sigma^{AE})^{\alpha_2} \nonumber\\
    &\leq p\max_{\cN:B\to E} \log \Tr (\sigma_1^{AE})^{\alpha_1}+(1-p) \max_{\cN:B\to E} \log \Tr (\sigma_2^{AE})^{\alpha_2} \nonumber\\
    &=pZ_{\alpha_1}(A:R)_{\rho}+(1-p)Z_{\alpha_2}(A:R)_{\rho}, \nonumber
\end{align}
where in the first line the state $\sigma^{AE}$ is the optimal state, and the third line follows from Lemma~\ref{lemma: g_m(alpha)}, i.e. $\alpha \mapsto \log \Tr (\sigma)^{\alpha}$ is convex. Hence, $Z_{\alpha}(A:R)_{\rho}$ is continuous at any $\alpha>1$.

\begin{comment}
Furthermore, note that $Z_{\alpha}(A:R)_{\rho}$ is decreasing with $\alpha$, that is for $\alpha_1<\alpha_2$ we obtain:
\begin{align}
    Z_{\alpha_2}(A:R)_{\rho}&= \log \Tr (\sigma_2^{AE})^{\alpha_2} \nonumber\\
    &\leq \log \Tr (\sigma_2^{AE})^{\alpha_1} \nonumber\\
    & \leq \max_{\cN:B\to E} \log \Tr (\sigma_1^{AE})^{\alpha_1}\nonumber\\
    &=Z_{\alpha_1}(A:R)_{\rho}, \nonumber
\end{align}
where in the first line $\sigma_2^{AE}$ is the optimal state, and the second line follows from Lemma~\ref{lemma: g_m(alpha)}, namely $\log \Tr (\sigma)^{\alpha}$ is decreasing in $\alpha$.
\end{comment}

To show the continuity from the right at $\alpha=1$, note that the R\'enyi entropy is continuous with $\alpha$, i.e. $\lim_{\alpha \to 1^+} S_{\alpha}(\rho)=S(\rho)$, as well as the quantum states. Also the minimization is over a compact set, therefore, 
we conclude it is also upper semi-continuous at $\alpha=1$, so they are continuous at $\alpha=1$ \cite[Thms.~10.1 and 10.2]{Rockafeller},
i.e. 
$\lim_{\alpha \to 1^+}E_{\alpha,p}(A:R)_{\rho}=E_{p}(A:R)_{\rho}$.
%\noindent 3. This follows from point 1. 

\noindent  (iii) In point (ii), we show that the function $\alpha \mapsto Z_{\alpha}(A:R)_{\rho}$ is   convex, therefore, its regularized version defined in Eq.~(\ref{eq: Z^infty}) is   convex therefore continuous at any $\alpha>1$
\begin{align}\label{eq: Z^infty}
     Z_{\alpha}^{\infty}(A:R)_{\rho}&:=\lim_{n\to \infty } \frac{Z_{\alpha}(A^n:R^n)_{\rho^{\otimes n}}}{n}. 
\end{align}
This implies that $E_{\alpha,p}^{\infty}(A:R)_{\rho}=\frac{1}{1-\alpha}  Z_{\alpha}^{\infty}(A:R)_{\rho}$ is continuous for $\alpha>1$.

\end{proof}

\medskip

To prove a strong converse bound, we use the following lemma where the fidelity between two states is upper bounded by the R\'enyi entropies of the corresponding states. This lemma, which was initially proved in~\cite{Dam2002}, was later generalized in \cite{Leditzky2016} to bound the fidelity by  conditional R\'enyi  entropies. 

\begin{lemma}[\cite{Dam2002}] \label{lemma: dam and Hayden}
For $\beta \in (\frac{1}{2},1)$ and any two states
$\rho$ and $\sigma$ on a finite-dimensional quantum system the following holds:
\begin{align*}
    S_{\beta}(\rho) \geq S_{\alpha}(\sigma)+\frac{2\beta}{1-\beta}F(\rho,\sigma),
\end{align*}
where $\alpha=\frac{\beta}{2 \beta -1}$.
\end{lemma}

\begin{theorem}\label{thm: F_v<S_beta} 
For any $\alpha>1$, the fidelity for the visible compression of mixed states with the rate $\cQ_v(n,\epsilon)$ is bounded as
\begin{align}
    \overline{F} \leq 2^{-n\frac{\alpha-1}{2\alpha}(E_{\alpha,p}^{\infty}(A:R)_{\rho}-\cQ_v(n,\epsilon))}. \nonumber
\end{align}
\end{theorem}

\begin{proof}
We bound the compression rate for $\beta \in (\frac{1}{2},1)$ as follows:
\begin{align} 
   n \cQ_v(n,\epsilon)&\geq S_{\beta}(M) \nonumber\\
        & = S_{\beta}(\hat{A}^nV)_{\xi} \nonumber\\
        &\geq S_{\alpha}(A^n V)_{\tau}+\frac{2\beta}{1-\beta} \log F(\tau^{A^n V },\xi^{\hat{A}^n V}) \nonumber\\
        &\geq S_{\alpha}(A^n V)_{\tau}+\frac{2\beta}{1-\beta} \log \overline{F} \label{eq: Q_v}
\end{align}
where the  equality  follows because the decoding isometry does not change the R\'enyi entropy.
The second inequality follows from Lemma~\ref{lemma: dam and Hayden} for $\alpha=\frac{\beta}{2\beta -1}$. This lemma holds for  arbitrary states $\tau$ and $\xi$. Here we let $\tau^{A^n V}=\sum_{x^n}p(x^n) \tau_{x^n}^{A^nV}$
where for all $x^n$ the state $\tau_{x^n}^{A^nV}$ is any extension of the state $\rho_{x^n}^{A^n}$ such that $F(\tau^{A^nV},\xi^{\hat{A}^nV})\geq \overline{F}$.
We can construct $\tau_{x^n}^{A^nV}$ as follows: let $\ket{\xi_{x^n}}^{\hat{A}^nV V'}$ and 
let $\ket{\psi_{x^n}}^{A^nV V'}$ be purifications of $\xi_{x^n}^{\hat{A}^nV}$ and $\rho_{x^n}^{A^n}$, respectively.
By Uhlmann's theorem, for any $x^n$ there is a unitary $U_{x^n}:VV'\to VV'$ such that 
\begin{align}
    F(\rho_{x^n}^{A^n},\xi_{x^n}^{\hat{A}^n})=F(\underbrace{(\1_{A^n} \otimes U_{x^n})\psi_{x^n}^{A^nV V'}(\1_{A^n} \otimes U_{x^n})^{\dagger}}_{\tau_{x^n}^{A^nVV'}},\xi_{x^n}^{\hat{A}^nV V'})
    \leq F(\tau_{x^n}^{A^nV},\xi_{x^n}^{\hat{A}^nV}), \nonumber
\end{align}
where the inequality  follows from the monotonicity of the fidelity under partial trace on system $V'$. By taking the average on the both sides of the above inequality, we obtain 
\begin{align}
   \overline{F}&=\sum_{x^n}p(x^n) F(\rho_{x^n}^{A^n},\xi_{x^n}^{\hat{A}^n}) \nonumber \\
   &\leq \sum_{x^n}p(x^n) F(\tau_{x^n}^{A^nV},\xi_{x^n}^{\hat{A}^nV}). \nonumber\\
   &\leq  F(\tau^{A^nV},\xi^{\hat{A}^nV}), \nonumber
\end{align}
where the last inequality follows from concavity of the fidelity.

Next, we bound the R\'enyi entropy $S_{\alpha}(A^n V)_{\tau}$ in Eq.~(\ref{eq: Q_v}) where we defined the state  $\tau_{x^n}^{A^nV}$ as  an extension of the state~$\rho_{x^n}^{A^n}$.
Consider a purification  $\ket{\psi}^{A^n {R'}^n {R''}^n R^n}=\sum_{x^n} \sqrt{p(x^n)} \ket{\psi_{x^n}}^{A^n {R'}^n}\ket{x^n}^{{R''}^n}\ket{x^n}^{R^n}$ of the source $\rho^{A^n R^n}$.
Then, any extended state $\tau^{A^nV {R''}^n}=\sum_{x^n}p(x^n) \tau_{x^n}^{A^nV} \otimes \ketbra{x^n}^{{R''}^n}$ is obtained by applying a CPTP map $\Lambda: {R'}^n {R''}^n\to V$ as follows:
$\tau^{A^nV}=(\id_{A^n}\otimes \Lambda)\rho^{A^n {R'}^n {R''}^n}$ where $\rho^{A^n {R'}^n {R''}^n} =\Tr_{R^n}\ketbra{\psi}^{A^n {R'}^n {R''}^n R^n}$.
Hence, we obtain the  inequality 
\begin{align} 
       S_{\alpha}(A^n V)_{\tau}
        &\geq \mathop{\min_{ \Lambda: {R'}^n {R''}^n \to E}} 
        S_{\alpha}(A^n E)_{\nu} \nonumber\\
        &= E_{p,\alpha}(A^n:R^n)_{\rho^{\otimes n}}, \label{eq: S_beta > E_p_beta}
\end{align}
where the entropy in the minimization is with respect to the state $\nu^{A^nE}:=(\id_{A^n}\otimes \Lambda)\rho^{A^n {R'}^n {R''}^n}$.
The second line is due to the definition of the R\'enyi entanglement of purification
for the state $({\rho}^{A R})^{\otimes n}$.

From Eq.~(\ref{eq: Q_v}) and Eq.~(\ref{eq: S_beta > E_p_beta}), we obtain
\begin{align} 
    n\cQ_v(n,\epsilon)&\geq \frac{1}{n}E_{\alpha,p}(A^n:R^n)_{\rho^{\otimes n}}+\frac{2\beta}{(1-\beta)n} \log \overline{F}  \nonumber \\
     &\geq E_{\alpha,p}^{\infty}(A:R)_{\rho}+\frac{2\beta}{(1-\beta)n} \log \overline{F}, \label{eq: Q-v>F-v}
\end{align}
where the second inequality follows because  the definition of the R\'enyi entanglement of purification implies that  $mE_{\alpha,p}(A^n:R^n)_{\rho^{\otimes n}}\geq E_{\alpha,p}(A^{nm}:R^{nm})_{\rho^{\otimes nm}}$ holds for any $m$. Hence, we have
\begin{align}
    \frac{1}{n}E_{\alpha,p}(A^n:R^n)_{\rho^{\otimes n}}\geq \lim_{m\to \infty}\frac{1}{nm}E_{\alpha,p}(A^{nm}:R^{nm})_{\rho^{\otimes nm}}=E_{\alpha,p}^{\infty}(A:R)_{\rho}, \nonumber
\end{align}
where the equality follows because every subsequence of a convergent sequence converges to the same limit as the original sequence. The theorem follows by rearranging the terms in Eq.~(\ref{eq: Q-v>F-v}).
\end{proof}

%where $\cN$ is CPTP map such that $\sigma^{AE}=(\id_A \otimes \cN)\rho^{AR}$, and $\rho^{AR}=\Tr_B \ketbra{\psi}^{ABR}$ for a purification $\ket{\psi}^{ABR}$ of $\rho^{AB}$.
In Lemma~\ref{lemma: E_alpha_p properties}, we show that ${E}_{\alpha,p}^{\infty}(A:R)_{\rho}$ is continuous for $\alpha>1$, however, it is not known if it is continuous at $\alpha=1$. %Therefore, we define $\widetilde{E}_{p}^{\infty}(A:B)_{\rho}:=\lim_{\alpha \to 1^+}E_{\alpha,p}^{\infty}(A:B)_{\rho}$ and 
We prove that for any rate below $\lim_{\alpha \to 1^+}E_{\alpha,p}^{\infty}(A:R)_{\rho}$ the fidelity converges to zero.

\begin{proposition}\label{proposition: F_v}
For any rate $\cQ_v(n,\epsilon)< \lim_{\alpha \to 1^+}E_{\alpha,p}^{\infty}(A:R)_{\rho}$, there is a $K>0$ such that the following  bound holds for the visible compression of mixed states: 
\begin{align}
    \overline{F} \leq 2^{-nK}. \nonumber
\end{align}
\end{proposition}
\begin{proof}
From Lemma~\ref{lemma: E_alpha_p properties} we know that the regularized R\'enyi entanglement of purification $E_{\alpha,p}^{\infty}(A:R)_{\rho}$ is continuous and decreasing with respect to $\alpha$ for $\alpha > 1$.
Therefore, for any $\cQ_v(n,\epsilon)< \lim_{\alpha \to 1^+}E_{\alpha,p}^{\infty}(A:R)_{\rho}$, there is a $\alpha>1$ such that $\cQ_v(n,\epsilon)<{E}_{\alpha,p}^{\infty}(A:R)_{\rho} < \lim_{\alpha \to 1^+}E_{\alpha,p}^{\infty}(A:R)_{\rho}$. Hence, the claim follows from Theorem~\ref{thm: F_v<S_beta} by letting
$K:=\frac{\alpha-1}{2\alpha}({E}_{\alpha,p}^{\infty}(A:R)_{\rho}-\cQ(n,\epsilon))$.
\end{proof}
\begin{remark}
If $\lim_{\alpha \to 1^+}E_{\alpha,p}^{\infty}(A:R)_{\rho}=E_{p}^{\infty}(A:R)_{\rho}$, then the above
bound is the strong converse for the visible compression of mixed states.
\end{remark}

\begin{remark}
The strong converse bound of the visible compression  holds for the blind compression as well, since the bound is obtained assuming that the encoder has access to a side information system. Namely, for any CPTP decoder and any rate $\cQ_b(n,\epsilon) < \lim_{\alpha \to 1^+}E_{\alpha,p}^{\infty}(A:R)_{\rho}$, there is a $K>0$ such that the following  bound holds for the visible compression of mixed states: 
\begin{align}
    \overline{F} \leq 2^{-nK}. \nonumber
\end{align}
For rates $\lim_{\alpha \to 1^+}E_{\alpha,p}^{\infty}(A:R)_{\rho} \leq \cQ_b(n,\epsilon) < S(CQ)_{\omega}$,
we obtain the above bound by imposing the assumption that the decoder is an isometry.

\end{remark}
%%%%%%%%%%%%%%%%%%%%%%%%%%%%%%%%%%%%%%%%%%%%%%%%%%%%%%%%%%%%%%%%%%%%%%%%%%%
\section{Blind compression case }\label{section: blind}
In this section, we consider  the compression of a general
mixed state source $\rho^{AR}$ where  $A$ is the system to be compressed, and  $R$
is the reference system. The ensemble source $\{p(x),\rho_x^A\}$ considered  by Koashi and Imoto in \cite{KI2001,KI2002} is a special case of $\rho^{AR}$ where the reference is a classical system
as follows:
\begin{align}\label{eq: KI source}
    \rho^{AR}=\sum_x p(x) \rho_x^{A} \otimes \ketbra{x}^R, 
\end{align}
and $\{\ket{x}^R\}_x$ is  the orthonormal basis for the classical system $R$. To be more specific about the role of the reference system, let the CPTP map $\cN: A \to \hat{A}$ denote the combined action of the encoder and the decoder. Then,
preserving the average fidelity, as considered by Koashi and Imoto, is equivalent to preserving the classical-quantum state $\rho^{AR}$ in Eq.~(\ref{eq: KI source}):
\begin{align}
    \sum_x p(x) F(\rho_x^{A}, \cN (\rho_x^{A}))=F(\rho^{AR},(\cN \otimes \id_R)\rho^{AR}). \nonumber
\end{align}
In this section, we consider a general mixed-state source $\rho^{AR}$, which is not necessarily a classical-quantum state as in~Eq.~(\ref{eq: KI source}), and prove a strong converse bound for it. This immediately implies a strong converse for the  source defined in Eq.~(\ref{eq: KI source}).

%We consider the same encoding-decoding model as in Section~\ref{sec: compression model}. 
%
%The encoder performs the encoding  isometry  $V_{\mathcal{C}}:A^n \hookrightarrow MW$ to obtain the compressed system $M$ and an environment system $W$:
%\begin{align}
%   (V_{\cC} \otimes \1_{R^n}) (\rho^{AR})^{\otimes n} (V_{\cC} \otimes \1_{R^n})^{\dagger}=:\sigma^{MWR^n}. \nonumber
%\end{align}
%Receiving $M$, the decoder performs the decoding isometry $V_{\mathcal{D}}:M \hookrightarrow \hat{A}^n V$ to obtain the reconstructed system $\hat{A}^n$ and the environment system $V$:
%\begin{align}
%   (V_{\cD} \otimes \1_{WR^n}) \sigma^{MWR^n} (V_{\cD} \otimes \1_{WR^n})^{\dagger}=:\xi^{\hat{A}^nWVR^n}. \nonumber
%\end{align}
%The fidelity for this  blind scheme is defined as follows:
%\begin{align}
 %   \overline{F}_b=F(\rho^{A^nR^n}, \xi^{\hat{A}^nR^n}), \label{eq: F_b}
%\end{align}
%where $\xi^{\hat{A}^nR^n}=\Tr_{VW} \xi^{\hat{A}^nWVR^n}$.

\medskip
As shown in \cite{general_mixed_state_compression,ZK_mixed_state_ISIT_2020,ZBK_PhD}, the optimal compression rate of a general mixed-state source $\rho^{AR}$ is equal to $S(CQ)_{\omega}$, i.e the von Neumann entropy of the classical and quantum systems in a decomposition 
of this state introduced in~\cite{Hayden2004}, which is a generalization of the decomposition
introduced by Koashi and Imoto in~\cite{KI2002}. Namely, for any set of quantum states $\{ \rho_x^A\}$, 
there is a unique decomposition of the Hilbert space describing
the structure of CPTP maps which preserve the set $\{ \rho_x^A\}$. This idea was generalized 
in \cite{Hayden2004} for a general mixed state $\rho^{AR}$ describing the structure of 
CPTP maps acting on system $A$ which preserve the overall state $\rho^{AR}$. 
This was achieved by showing that any such map preserves the set of all possible
states on system $A$ which can be obtained by measuring system $R$, and 
conversely any map preserving the set of all possible
states on system $A$ obtained by measuring system $R$, preserves the state $\rho^{AR}$,
thus reducing the general case to the case of classical-quantum states 
\begin{align*}
    \rho^{AY} = \sum_y q(y) \rho_y^A \otimes \proj{y}^Y = \sum_y \Tr_R \rho^{AR}(\1_A\otimes M_y^R) \otimes \proj{y}^Y, 
\end{align*}
which is the ensemble case considered by Koashi and Imoto. 
%
%As a matter of fact,  looking at the algorithm presented in \cite{KI2002} to compute the decomposition, it is enough to consider an informationally complete POVM $(M_y)$ on $R$, with no more than $|R|^2$ many outcomes, which are measurements $(M_y)$ on $R$ such that the map $\mathcal{M}:\rho \to \mathcal{M}(\rho) = \Tr(\rho M_y)$ from states to probability distributions is one-to-one, or equivalently, such that the operators $M_y$ of the measurement spans all operators on $R$, i.e. $\operatorname{span}\{M_y\} = \mathcal{B}(R)$ \cite{Informationally_complete_POVM2006}.
%
The properties of this decomposition are stated in the following theorem.

\begin{theorem}[\cite{KI2002,Hayden2004}]
\label{thm: KI decomposition}
Associated to the state $\rho^{AR}$, there are Hilbert spaces $C$, $N$ and $Q$
and an isometry $U_{\KI}:A \hookrightarrow C N Q$ such that:
\begin{enumerate}[(i)]
  \item The state $\rho^{AR}$ is transformed by $U_{\KI}$ as
    \begin{align}
        \label{eq:KI state}
      (U_{\KI}\!\otimes\! \1_R)\rho^{AR} (U_{\KI}^{\dagger}\!\otimes\! \1_R)
        \!\!= \!\!\sum_j p_j \proj{j}^{C}\!\! \otimes \omega_j^{N} \otimes \rho_j^{Q R} 
        =:\omega^{C N Q R},
    \end{align}
    where the set of vectors $\{ \ket{j}^{C}\}$ form an orthonormal basis for Hilbert space 
    $C$, and $p_j$ is a probability distribution over $j$. The states $\omega_j^{N}$ and 
    $\rho_j^{Q R}$ act  on the Hilbert spaces $N$ and $Q \otimes R$, respectively.

  \item For any CPTP map $\Lambda$ acting on system $A$ which leaves the state $\rho^{AR}$ 
    invariant, that is $(\Lambda \otimes \id_R )\rho^{AR}=\rho^{AR}$, every associated 
    isometric extension $U: A\hookrightarrow A E$ of $\Lambda$ with the environment system 
    $E$ is of the following form
    \begin{equation}
      U = (U_{\KI}\otimes \1_E)^{\dagger}
            \left( \sum_j \proj{j}^{C} \otimes U_j^{N} \otimes \1^{Q} \right) U_{\KI},
    \end{equation}
    where the isometries $U_j:N \hookrightarrow N E$ satisfy 
    $\Tr_E [U_j \omega_j U_j^{\dagger}]=\omega_j$ for all $j$.
    The isometry $U_{KI}$ is unique (up to a trivial change of basis of the Hilbert spaces 
    $C$, $N$ and $Q$). Henceforth, we call the isometry $U_{\KI}$ and the state 
    $\omega^{C N Q R}$ defined by Eq.~(\ref{eq:KI state})
    the Koashi-Imoto (KI) isometry and KI-decomposition of the state $\rho^{AR}$, respectively.

  \item In the particular case of a tripartite system $CNQ$ and a state $\omega^{CNQR}$ already 
    in Koashi-Imoto form (\ref{eq:KI state}), property 2 says the following:
    For any CPTP map $\Lambda$ acting on systems $CNQ$ with 
    $(\Lambda \otimes \id_R )\omega^{CNQR}=\omega^{CNQR}$, every associated 
    isometric extension $U: CNQ\hookrightarrow CNQ E$ of $\Lambda$ with the environment system 
    $E$ is of the form
    \begin{equation}
      U = \sum_j \proj{j}^{C} \otimes U_j^{N} \otimes \1^{Q},
    \end{equation}
    where the isometries $U_j:N \hookrightarrow N E$ satisfy 
    $\Tr_E [U_j \omega_j U_j^{\dagger}]=\omega_j$ for all $j$.
\end{enumerate} 
\end{theorem}
In the following theorem, we  obtain a bound on the fidelity for the compression of a $\delta$-typical state $\omega^{C^nQ^nR^n}_{\delta}$, which is a Koashi-Imoto decomposition of $\rho^{A^nR^n}$. 
The Koashi-Imoto decomposition of $\rho^{\otimes n}$ is equal to the tensor product of Koashi-Imoto states i.e., $\omega^{\otimes n}$ \cite{KI2001,KI2002}. Consider the following state which is obtained by applying the Koashi-Imoto isometry 
on the decoded system $\hat{A}^n$:
\begin{align}\label{eq: state zeta}
   \zeta^{\hat{C}^n \hat{N}^n \hat{Q}^n  R^n}:= (U_{\KI}^{\otimes n}\otimes \1_{R^n}) \xi^{\hat{A}^n  R^n} (U_{\KI}^{\otimes n}\otimes \1_{R^n})^{\dagger}.
%   \zeta^{\hat{C}^n \hat{N}^n \hat{Q}^n WV R^n}:= (U_{\KI}^{\otimes n}\otimes \1_{R^n}) \xi^{\hat{A}^n WV R^n} (U_{\KI}^{\otimes n}\otimes \1_{R^n})^{\dagger}.
\end{align}
For a $\delta$-typical state $\omega^{C^nQ^n}_{\delta}=\Pi_{\delta} \omega^{C^nQ^n} \Pi_{\delta}$, the operator norm
is bounded by  $2^{-nS(CQ)_{\omega}+n\delta}$ (see Sec.~\ref{sec: appendix}). We use this and apply H\"{o}lder's inequality to bound the fidelity.

%P(\omega^{C^nQ^nR^n}_{\delta},\zeta^{\hat{C}^n\hat{Q}^nR^n})=\sqrt{1-F^2(\omega^{C^nQ^nR^n}_{\delta},\zeta^{\hat{C}^n\hat{Q}^nR^n})}

 %$\overline{F}$ which is with respect to the R\'enyi entropy. Below, in Proposition~\ref{proposition: F_b} 
%we state the strong converse bound.

\begin{theorem}\label{thm: F_b<S_beta} 
%For any $\alpha>1$, assuming that the decoder is an isometry, 
For any super-unital decoder, the fidelity for the blind compression of a $\delta$-typical state $\omega^{C^nQ^nR^n}_{\delta}=(\Pi_{\delta}\ox \1_{R^n})\omega^{C^nQ^nR^n} (\Pi_{\delta}\ox \1_{R^n})$, with $\omega^{C^nQ^nR^n}$ as a Koashi-Imoto decomposition of $\rho^{A^nR^n}$, is bounded as 
%mixed states with the rate $\cQ_b(n,\epsilon)$ is bounded as:
\begin{align}
F^2(\omega^{C^nQ^nR^n}_{\delta},\zeta^{\hat{C}^n\hat{Q}^nR^n})\leq  2^{-nS(CQ)_{\omega}+n\delta+n\cQ(n,\epsilon)},  
\end{align}
%\begin{align}
%\cM(n,\epsilon)\geq nS(CQ)_{\omega}-n\delta+2\log (1-\epsilon^2)
   % \overline{F} \leq 
    %2^{-n\frac{\alpha-1}{2\alpha}(S_{\alpha}(CQ)_{\omega}-Q_b)}. \nonumber
%\end{align}
\end{theorem}

\begin{proof}
Consider the $\delta$-typical state $\omega^{C^nQ^n}_{\delta}$ as defined in Section~\ref{sec: appendix}.
We obtain
\begin{align}
F(\omega^{C^nQ^nR^n}_{\delta},\zeta^{\hat{C}^n\hat{Q}^nR^n})&\leq F(\omega^{C^nQ^n}_{\delta},\zeta^{\hat{C}^n\hat{Q}^n}) \nonumber\\
&=\norm{\sqrt{\omega^{C^nQ^n}_{\delta}} \sqrt{\zeta^{\hat{C}^n\hat{Q}^n}}}_1 \nonumber\\
&\leq \norm{\sqrt{\omega^{C^nQ^n}_{\delta}} }_{\infty} \norm{\sqrt{\zeta^{\hat{C}^n\hat{Q}^n}}}_1 \nonumber\\
&\leq 2^{\frac{-nS(CQ)_{\omega}+n\delta}{2}} \norm{\sqrt{\zeta^{\hat{C}^n\hat{Q}^n}}}_1 \nonumber\\
&\leq 2^{\frac{-nS(CQ)_{\omega}+n\delta}{2}} \norm{\sqrt{\zeta^{\hat{C}^n\hat{Q}^n}}}_2 \norm{\sqrt{\1_{\hat{C}^n\hat{Q}^n}}}_2 \nonumber\\
&= 2^{\frac{-nS(CQ)_{\omega}+n\delta}{2}}  \norm{\sqrt{\1_{\hat{C}^n\hat{Q}^n}}}_2 \nonumber\\
&\leq 2^{\frac{-nS(CQ)_{\omega}+n\delta}{2}}  \norm{\sqrt{\cD(\1_{M})}}_2 \nonumber\\
&= 2^{\frac{-nS(CQ)_{\omega}+n\delta}{2}}   \sqrt{\tr(\cD(\1_{M}))} \nonumber\\
&= 2^{\frac{-nS(CQ)_{\omega}+n\delta+\log |M|}{2}},  
\end{align}
where the first line is due to monotonicity of the fidelity under partial trace operation.
In the third and fifth lines, we apply H\"{o}lder's inequality. In the seventh line, we assume that the decoder
is super-unital i.e. $\1_{\hat{C}^n\hat{Q}^n} \leq \cD(\1_{M})$. The last line follows because the decoder is a trace preserving map.

%\cM(n,\epsilon)\geq nS(CQ)_{\omega}-n\delta+2\log (1-\epsilon^2)
\end{proof}

We now obtain a bound on the fidelity for the compression of the original state.
Applying an isometry or a reversible map does not change the fidelity 
\begin{align}
    \overline{F}=F(\rho^{A^n R^n},\xi^{\hat{A}^nR^n}) =F(\omega^{C^n N^n Q^n R^n}, \zeta^{\hat{C}^n \hat{N}^n \hat{Q}^n R^n}) =F(\omega^{C^n Q^n R^n}, \zeta^{\hat{C}^n  \hat{Q}^n R^n}). \nonumber
\end{align}
By Lemma~\ref{lemma: typicality} and inequality in Eq.~(\ref{lemma: Fuchs}), the fidelity $F(\omega^{C^n N^n Q^n R^n}, \omega^{C^n N^n Q^n R^n}_{\delta})\geq \sqrt{1-2^{-nc \delta^2}}$
for some constant $c > 0$ independent of $\delta$ and $n$. 
In Theorem~\ref{thm: F_b<S_beta}, we bounded the fidelity for $\omega^{C^n Q^n R^n}_{\delta}$. 
Now, we apply the fidelity relation in Eq.~(4) of \cite{Barnum1996} to obtain
\begin{align}
\overline{F}&=F(\omega^{C^n Q^n R^n}, \zeta^{\hat{C}^n  \hat{Q}^n R^n}) \nonumber\\
& \leq  2^{\frac{-nS(CQ)_{\omega}+n\delta+\log |M|}{2}} + 2\left(1 - ({1-2^{-nc \delta^2}})^{\frac{1}{4}}\right) + 2\sqrt{2} \sqrt{ 2^{\frac{-nS(CQ)_{\omega}+n\delta+\log |M|}{2}} \left(1 - ({1-2^{-nc \delta^2}})^{\frac{1}{4}}\right)}.  
\end{align}
The above fidelity converges to 0 as $n$ goes to $\infty$. 

%\begin{align}
%P(\omega^{C^nQ^nR^n},\zeta^{\hat{C}^n\hat{Q}^nR^n})\leq P(\omega^{C^nQ^nR^n}_{\delta},\zeta^{\hat{C}^n\hat{Q}^nR^n})+P(\omega^{C^nQ^nR^n},\omega^{C^nQ^nR^n}_{\delta})
%\end{align}

\section{Discussion}\label{sec: discussion}
The entanglement of purification appears as the optimal rate in the visible compression of mixed states. 
We show that it is additive for partially-exchangeable states, which exhibit certain extendibility  properties.
We took inspiration from the additivity proof and extended it to establish a lower bound
on the compression rate.  We obtain the lower bound in terms of smooth min-max entropies, which leave room for handling
larger errors, as opposed to continuity bounds for the von Neumann entropy, which typically lead to weak converse bounds in the regime of  vanishing error.
%lead to meaningful lower bounds only in the case of vanishing error.
%
In the lower bound, we apply Lemma~\ref{lemma: H_min < H_max+terms} as well as  the chain rules of Lemma~\ref{lemma: smooth chain rules} three times. This restricts
the range of error up to $\frac{1}{3\sqrt{2}}$. This is similar to the approach in \cite{AW_pretty}, which derives a pretty converse bound for the quantum and private capacity of degradable channels. However, in that case, the chain rule is applied only once, allowing for larger errors up to $\frac{1}{\sqrt{2}}$.

Then, to deal with the visible compression of more general mixed states, we define a new quantity in Definition~\ref{def: E_alpha_p} which is similar to the entanglement of purification, however, instead of von Neumann entropy, the minimization of the  R\'enyi entropy is taken into account. We obtain an upper bound on the fidelity criterion in terms of the regularization of the $\alpha$-R\'enyi entanglement of purification  $E^{\infty}_{\alpha,p}(A:R)_{\rho}$ and show that for any rate below $\lim_{\alpha \to 1^+}E^{\infty}_{\alpha,p}(A:R)_{\rho}$
the fidelity exponentially converges to zero.
However, it remains an open question whether this quantity is continuous, namely, if $\lim_{\alpha \to 1^+}E_{\alpha,p}^{\infty}(A:R)_{\rho}=E_{p}^{\infty}(A:R)_{\rho}$ holds true? 
This continuity obviously implies that the strong converse bound holds for the visible compression scheme.

Finally, we prove a strong converse bound for the compression of a general mixed-state source $\rho^{AR}$ shared between an encoder and an inaccessible reference system as defined in \cite{ZBK_PhD,ZK_mixed_state_ISIT_2020,general_mixed_state_compression} by  assuming that the decoder is a super-unital map. This implies a strong converse for the blind compression of ensembles of mixed states, by assuming a super-unital decoder, since it is a special case of the general mixed-state source when the reference system $R$ is classical. 
The proof is relatively simple: We apply H\"{o}lder's inequality  to bound the fidelity for states supported on the typical subspace, where the logarithm of the operator norm is approximately bounded by the entropy of the non-redundant parts of the source. Then, we use the fact that for any
super-unital channel $\cD:A \to B$ the inequality  $\1_B \leq \cD(\1_A)$ holds.

\bigskip

\noindent \textbf{Acknowledgments.}
%I am grateful to  Andreas Winter and Debbie Leung for helpful discussions and comments on this manuscript. 
%Robert K\"onig for helpful discussions and comments on this manuscript and to Chokri Manai for discussions regarding continuity issues in Lemma~\ref{lemma: g_m(alpha)}. 
%I also thank Kohdai Kuroiwa for pointing out an error in Lemma 14 in the first version of the paper.
%
%
I am grateful to Andreas Winter and Debbie Leung for their helpful discussions and to Kohdai Kuroiwa for pointing out an error in Lemma 14 in the first version of this paper.
The author was supported by the MCQST (Munich Center for Quantum Science and Technology) Distinguished Postdoctoral Fellowship
and the Ada Lovelace Postdoctoral Fellowship at Perimeter Institute for Theoretical Physics.
%
%by the DFG cluster of excellence 2111 (Munich Center for Quantum Science and Technology) 
%%%%%%%%%%%%%%%%%%%%%%%%%%%%%%%%%%%%%%%%%%%%%%%%%%%%%%%%%%%%%%%%%%%%%%%%%%%%%%%%%%%%%
\appendix

\section*{Appendix: Miscellaneous Lemmas}\label{sec: appendix}
In this section, we include various lemmas and facts that we apply throughout the paper.

\medskip
%In this section, we complete the proof of the coding theorem by applying the one-shot result (Theorem II) of the previous section to channels of the form $\mathcal{N}^{\otimes n}$, obtaining, for any $\varphi^{A'}$, codes achieving rates arbitrarily close to $I_c(\varphi^{A'}, \mathcal{N})$. The rough idea is that it will be possible to replace the quantities appearing on the right-hand side of (3) by entropic quantities because of the memoryless structure of the channel. As a first step to making this idea precise, we begin by recalling some needed ideas from the method of types. 
Consider a density matrix with spectral decomposition $\varphi^{A'} = \sum_x p_x \proj{x}^{A'}$. Its $n$-tensor power can be written as
\begin{equation*}
    (\varphi^{A'})^{\otimes n} = \sum_{x^n} p_{x^n}\proj{x^n}^{A^n}
\end{equation*}
where $p_{x^n} = p_{x_1} \cdots p_{x_n}$ and $|x^n\rangle^{A^n} = |x_1\rangle \otimes \cdots \otimes |x_n\rangle^{A}$. The $\delta$-(entropy) typical subspace $A_\delta \subseteq A^n$ is defined as
\begin{equation*}
    A_\delta = \text{span} \left\{ |x^n\rangle^{A^n} : \left| \frac{1}{n} \log p_{x^n} - S(\varphi^{A'}) \right| \leq \delta \right\}
\end{equation*}
and the $\delta$-typical projection $\Pi_\delta^{A^n}$ is defined to project $A^n$ onto $A_\delta$.

\begin{lemma}\label{lemma: typicality}(Typicality \cite{Horodecki2007}) \textit{Let a  state $\varphi^{A}$ be given. For every $\delta > 0$ and all sufficiently large $n$, there exist $\delta$-typical projections $\Pi_\delta^{(A)}$ onto $\delta$-typical subspaces $A_\delta \subseteq A^n$ } %and $B_\delta \subseteq B^n$} %, and $E_\delta \subseteq E^n$ such that the states}
\begin{align}
    \varphi^{A^n } &= (\varphi^{A})^{\otimes n}  \\
    \varphi_\delta^{A^n } &= \Pi_\delta^{A^n}   \varphi^{A^n }   \Pi_\delta^{A^n}  
\end{align}
satisfy
\begin{align}
    |A_\delta| &\leq 2^{nS(A)+n\delta} \\
    \tr \left( \Pi_\delta^{A^n} \right) &\leq 2^{-nS(A)+n\delta} \\
    \norm{ \varphi_\delta^{A^n }}_{\infty}&\leq 2^{-nS(A)+n\delta} \\
    \|\varphi^{A^n } - \varphi_\delta^{A^n }\|_1 &\leq \epsilon 
    \end{align}
where $\epsilon = 2^{-nc\delta^2}$ for some constant $c > 0$ independent of $\delta$ and $n$.
\end{lemma}

\begin{lemma}\label{lemma: smooth chain rules} (Chain rules \cite{One_Shot_Decoupling2014,smooth_chain_rules})
Let $\epsilon, \delta \geq 0$, $\eta > 0$.
Then, with respect to a state $\rho \in \cS(ABC)$,
\begin{align}
H_{\max}^{\epsilon}(AB|C)_{\rho} &\geq H_{\min}^{\delta}(B|C)_{\rho}+H_{\max}^{\epsilon+2\delta+2\eta}(A|BC)_{\rho}-3 \log\frac{2}{\eta^2}.\\
 H_{\max}^{\epsilon+2\delta+\eta}(AB|C)_{\rho}& \leq H_{\max}^{\epsilon}(A|BC)_{\rho}+H_{\max}^{\delta}(B|C)_{\rho}+\log\frac{2}{\eta^2}.
\end{align}

%\begin{align}
%H_{\max}^{\epsilon+2\delta+\eta}(AB|C)_{\rho} \leq H_{\min}^{\delta}(B|C)_{\rho}+H_{\max}^{\epsilon}(A|BC)_{\rho}+ \log\frac{2}{\eta^2}.
%\end{align}

\begin{lemma}\label{lemma: H_min < H_max+terms} (Proposition 5.5 in \cite{Tomamichel_PhD})
    Let $\rho \in \cS(AB)$, $\alpha,\beta \geq 0$ and $\alpha+\beta  <\frac{\pi}{2}$. Then
    \begin{align}
        H_{\min}^{\sin \alpha}(A|B)_{\rho} &\leq H_{\max}^{\sin \beta}(A|B)_{\rho}+\log \frac{1}{\cos^2(\alpha+\beta)}.  
            \end{align}
\end{lemma}

\begin{theorem}\label{thm: AEP}(Min- and max-entropy AEP \cite{Renner_PhD}, \cite{Tomamichel_PhD})
Let $\rho^{AB} \in \cS(AB)$ and $0\leq \epsilon \leq 1$. Then,
\begin{align}
  \lim_{n \to \infty} \frac{1}{n} H_{\max}^{\epsilon}(A|B)_{\rho^{\ox n}}=\lim_{n \to \infty} \frac{1}{n} H_{\min}^{\epsilon}(A|B)_{\rho^{\ox n}}=S(A|B)_{\rho}.
\end{align}
More precisely, for a purification $\ket{\psi}^{ABC}$ of $\rho$, denote
$\theta_X:=\log \norm{(\psi^X)^{-1}}$, where the inverse is the generalized
inverse (restricted to the support), for $X = B,C$. Then, for
every $n$,
\begin{align}
    H_{\max}^{\epsilon}(A|B)_{\rho^{\ox n}} &\leq nS(A|B)_{\rho} +(\theta_B+\theta_C)\sqrt{n \log \frac{2}{\epsilon}} \\
    H_{\min}^{\epsilon}(A|B)_{\rho^{\ox n}} &\geq nS(A|B)_{\rho} -(\theta_B+\theta_C)\sqrt{n \log \frac{2}{\epsilon}}.
    \end{align}
\end{theorem}

\end{lemma}

%We apply the following lemma to discuss continuity for the $\alpha$-R\'enyi entanglement of purification and its regularization.   
\begin{lemma}\label{lemma: g_m(alpha)}
Let $\{\sigma^n=\cT_n(\rho^{\otimes n})\}_n$ be a sequence of states for $n\geq 1$ where $\{\cT_n\}_n$ is a sequence of CPTP maps. For $\alpha>1$ define $g_{\alpha,n}(\sigma^n):=\frac{(1-\alpha)S_{\alpha}(\sigma^n)}{n}=\frac{\log \Tr(\sigma^n)^{\alpha} }{n}$. Then, $g_{\alpha,n}(\sigma^n)$ has the following properties  for $\alpha>1$:
\begin{enumerate}[(i)]
    \item It is a decreasing function of $\alpha$.
    \item It is a convex function of $\alpha$.
%    \item It is a continuous function of $\alpha$.
\end{enumerate}
\end{lemma}

\begin{proof}
(i) Consider the spectral decomposition $\sigma^n=\sum_{i=1}^{d^n} \lambda_{i,n}\ketbra{v_{i,n}}$ where $d$ is the dimension of a single output system.
Define $\gamma_{\alpha,n}:=\sum_{i=1}^{d^n} \lambda_{i,n}^{\alpha}$. Then, the derivative of $g_{\alpha,n}(\sigma^n)$ with respect to $\alpha$ is always non-positive:
\begin{align}
 g'_{\alpha,n}(\sigma^n)&=\frac{1}{n} \frac{\sum_{i=1}^{d^n} \lambda_{i,n}^{\alpha}\log\lambda_{i,n}}{\sum_{i=1}^{d^n} \lambda_{i,n}^{\alpha}} \nonumber\\ &=\frac{1}{n} \sum_{i=1}^{d^n} \frac{\lambda_{i,n}^{\alpha}}{\gamma_{\alpha,n}}\log\lambda_{i,n} \nonumber\\
 &\leq \frac{1}{n} \log \sum_{i=1}^{d^n} \frac{\lambda_{i,n}^{\alpha}}{\gamma_{\alpha,n}}\lambda_{i,n} \nonumber \\
 &\leq \frac{1}{n} \log \sum_{i=1}^{d^n} \lambda_{i,n}\nonumber \\
 &=0,\nonumber
\end{align}
where in the first line the denominator is equal to $\gamma_{\alpha,n}$. In the second line note that $\{\frac{\lambda_{i,n}^{\alpha}}{\gamma_{\alpha,n}}\}_i$ is a probability distribution, therefore, the the third line follows from the concavity of $\log(\cdot)$.
The penultimate line follows because $0\leq \frac{\lambda_{i,n}^{\alpha}}{\gamma_{\alpha,n}}\leq 1$. The last line follows because the eigenvalues sum up to one.

\noindent (ii) The second derivative with respect to $\alpha$ is always non-negative:
\begin{align}
    g''_{\alpha,n}(\sigma^n)&=\frac{1}{n} \frac{\sum_{i=1}^{d^n} \lambda_{i,n}^{\alpha}(\log\lambda_{i,n})^2}{\sum_{i=1}^{d^n} \lambda_{i,n}^{\alpha}}
    -\frac{1}{n} \frac{(\sum_{i=1}^{d^n} \lambda_{i,n}^{\alpha}\log\lambda_{i,n})^2}{(\sum_{i=1}^{d^n} \lambda_{i,n}^{\alpha})^2}\nonumber \\
    &=\frac{1}{n} \sum_{i=1}^{d^n} \frac{\lambda_{i,n}^{\alpha}}{\gamma_{\alpha,n}}(\log\lambda_{i,n})^2
    -\frac{1}{n} \left(\sum_{i=1}^{d^n} \frac{\lambda_{i,n}^{\alpha}}{\gamma_{\alpha,n}}\log\lambda_{i,n}\right)^2\nonumber\\
    & \geq \frac{1}{n} \sum_{i=1}^{d^n} \frac{\lambda_{i,n}^{\alpha}}{\gamma_{\alpha,n}}(\log\lambda_{i,n})^2-\frac{1}{n} \sum_{i=1}^{d^n} \frac{\lambda_{i,n}^{\alpha}}{\gamma_{\alpha,n}}(\log\lambda_{i,n})^2\nonumber\\
    &=0,\nonumber
\end{align}
where in the second line $\{\frac{\lambda_{i,n}^{\alpha}}{\gamma_{\alpha,n}}\}_i$ is a probability distribution, therefore, the third line follows from the concavity of $f(x)=-x^2$.

%\noindent (iii) The continuity follows because $g_{\alpha,n}(\sigma^n)$ is decreasing and convex for $\alpha>1$.
\end{proof} 

\vspace{-1cm}

%\begin{lemma}
%Let $\rho$ be unbounded density matrix. Define the following quantity:
%\begin{align*}
%    g_{\alpha}(\rho):=(\alpha-1) S_{\alpha}(\rho)=-\log \Tr \sum_{i=1}^m \lambda_i^{\alpha}.
%\end{align*}
%Then, $g_{\alpha}(\rho)$ has the follows properties:
%\begin{enumerate}
%    \item It is increasing in $\alpha$.
%    \item It is a concave in $\alpha$.
%    \item It is continuous in $\alpha$.
%\end{enumerate}
%
%$g_{\alpha}(\rho)$ is a concave and increasing function of $\alpha$
%\end{lemma}

\bibliography{strong_converse_references}

\begin{thebibliography}{34}%
\makeatletter
\providecommand \@ifxundefined [1]{%
 \@ifx{#1\undefined}
}%
\providecommand \@ifnum [1]{%
 \ifnum #1\expandafter \@firstoftwo
 \else \expandafter \@secondoftwo
 \fi
}%
\providecommand \@ifx [1]{%
 \ifx #1\expandafter \@firstoftwo
 \else \expandafter \@secondoftwo
 \fi
}%
\providecommand \natexlab [1]{#1}%
\providecommand \enquote  [1]{``#1''}%
\providecommand \bibnamefont  [1]{#1}%
\providecommand \bibfnamefont [1]{#1}%
\providecommand \citenamefont [1]{#1}%
\providecommand \href@noop [0]{\@secondoftwo}%
\providecommand \href [0]{\begingroup \@sanitize@url \@href}%
\providecommand \@href[1]{\@@startlink{#1}\@@href}%
\providecommand \@@href[1]{\endgroup#1\@@endlink}%
\providecommand \@sanitize@url [0]{\catcode `\\12\catcode `\$12\catcode
  `\&12\catcode `\#12\catcode `\^12\catcode `\_12\catcode `\%12\relax}%
\providecommand \@@startlink[1]{}%
\providecommand \@@endlink[0]{}%
\providecommand \url  [0]{\begingroup\@sanitize@url \@url }%
\providecommand \@url [1]{\endgroup\@href {#1}{\urlprefix }}%
\providecommand \urlprefix  [0]{URL }%
\providecommand \Eprint [0]{\href }%
\providecommand \doibase [0]{http://dx.doi.org/}%
\providecommand \selectlanguage [0]{\@gobble}%
\providecommand \bibinfo  [0]{\@secondoftwo}%
\providecommand \bibfield  [0]{\@secondoftwo}%
\providecommand \translation [1]{[#1]}%
\providecommand \BibitemOpen [0]{}%
\providecommand \bibitemStop [0]{}%
\providecommand \bibitemNoStop [0]{.\EOS\space}%
\providecommand \EOS [0]{\spacefactor3000\relax}%
\providecommand \BibitemShut  [1]{\csname bibitem#1\endcsname}%
\let\auto@bib@innerbib\@empty
%</preamble>
\bibitem [{\citenamefont {Schumacher}(1995)}]{Schumacher1995}%
  \BibitemOpen
  \bibfield  {author} {\bibinfo {author} {\bibfnamefont {B.}~\bibnamefont
  {Schumacher}},\ }\bibfield  {title} {\enquote {\bibinfo {title} {Quantum
  coding},}\ }\href@noop {} {\bibfield  {journal} {\bibinfo  {journal} {Phys.
  Rev. A}\ }\textbf {\bibinfo {volume} {51}},\ \bibinfo {pages} {2738--2747}
  (\bibinfo {year} {1995})}\BibitemShut {NoStop}%
\bibitem [{\citenamefont {Jozsa}\ and\ \citenamefont
  {Schumacher}(1994)}]{Jozsa1994_1}%
  \BibitemOpen
  \bibfield  {author} {\bibinfo {author} {\bibfnamefont {R.}~\bibnamefont
  {Jozsa}}\ and\ \bibinfo {author} {\bibfnamefont {B.}~\bibnamefont
  {Schumacher}},\ }\bibfield  {title} {\enquote {\bibinfo {title} {A new proof
  of the quantum noiseless coding theorem},}\ }\href@noop {} {\bibfield
  {journal} {\bibinfo  {journal} {J. Mod. Optics}\ }\textbf {\bibinfo {volume}
  {41}},\ \bibinfo {pages} {2343--2349} (\bibinfo {year} {1994})}\BibitemShut
  {NoStop}%
\bibitem [{\citenamefont {Horodecki}(2000)}]{Horodecki2000}%
  \BibitemOpen
  \bibfield  {author} {\bibinfo {author} {\bibfnamefont {M.}~\bibnamefont
  {Horodecki}},\ }\bibfield  {title} {\enquote {\bibinfo {title} {Optimal
  compression for mixed signal states},}\ }\href@noop {} {\bibfield  {journal}
  {\bibinfo  {journal} {Phys. Rev. A}\ }\textbf {\bibinfo {volume} {61}},\
  \bibinfo {pages} {052309} (\bibinfo {year} {2000})}\BibitemShut {NoStop}%
\bibitem [{\citenamefont {Horodecki}(1998)}]{Horodecki1998}%
  \BibitemOpen
  \bibfield  {author} {\bibinfo {author} {\bibfnamefont {M.}~\bibnamefont
  {Horodecki}},\ }\bibfield  {title} {\enquote {\bibinfo {title} {Limits for
  compression of quantum information carried by ensembles of mixed states},}\
  }\href@noop {} {\bibfield  {journal} {\bibinfo  {journal} {Phys. Rev. A}\
  }\textbf {\bibinfo {volume} {57}},\ \bibinfo {pages} {3364--3369} (\bibinfo
  {year} {1998})}\BibitemShut {NoStop}%
\bibitem [{\citenamefont {Hayashi}(2006)}]{visible_Hayashi}%
  \BibitemOpen
  \bibfield  {author} {\bibinfo {author} {\bibfnamefont {Masahito}\
  \bibnamefont {Hayashi}},\ }\bibfield  {title} {\enquote {\bibinfo {title}
  {Optimal visible compression rate for mixed states is determined by
  entanglement of purification},}\ }\href@noop {} {\bibfield  {journal}
  {\bibinfo  {journal} {Phys. Rev. A}\ }\textbf {\bibinfo {volume} {73}},\
  \bibinfo {pages} {060301} (\bibinfo {year} {2006})}\BibitemShut {NoStop}%
\bibitem [{\citenamefont {Koashi}\ and\ \citenamefont {Imoto}(2001)}]{KI2001}%
  \BibitemOpen
  \bibfield  {author} {\bibinfo {author} {\bibfnamefont {M.}~\bibnamefont
  {Koashi}}\ and\ \bibinfo {author} {\bibfnamefont {N.}~\bibnamefont {Imoto}},\
  }\bibfield  {title} {\enquote {\bibinfo {title} {Compressibility of quantum
  mixed-state signals},}\ }\href@noop {} {\bibfield  {journal} {\bibinfo
  {journal} {Phys. Rev. Lett.}\ }\textbf {\bibinfo {volume} {87}},\ \bibinfo
  {pages} {017902} (\bibinfo {year} {2001})}\BibitemShut {NoStop}%
\bibitem [{\citenamefont {Koashi}\ and\ \citenamefont {Imoto}(2002)}]{KI2002}%
  \BibitemOpen
  \bibfield  {author} {\bibinfo {author} {\bibfnamefont {M.}~\bibnamefont
  {Koashi}}\ and\ \bibinfo {author} {\bibfnamefont {N.}~\bibnamefont {Imoto}},\
  }\bibfield  {title} {\enquote {\bibinfo {title} {Operations that do not
  disturb partially known quantum states},}\ }\href@noop {} {\bibfield
  {journal} {\bibinfo  {journal} {Phys. Rev. A}\ }\textbf {\bibinfo {volume}
  {66}},\ \bibinfo {pages} {022318} (\bibinfo {year} {2002})}\BibitemShut
  {NoStop}%
\bibitem [{\citenamefont {B.~Khanian}(2020)}]{ZBK_PhD}%
  \BibitemOpen
  \bibfield  {author} {\bibinfo {author} {\bibfnamefont {Z.}~\bibnamefont
  {B.~Khanian}},\ }\emph {\bibinfo {title} {From Quantum Source Compression to
  Quantum Thermodynamics}},\ \href@noop {} {\bibinfo {type} {{PhD} thesis}},\
  \bibinfo  {school} {Universitat Aut\`{o}noma de Barcelona, Department of
  Physics}, \bibinfo {address} {Spain} (\bibinfo {year} {2020}),\ \bibinfo
  {note} {arXiv:quant-ph/2012.14143}\BibitemShut {NoStop}%
\bibitem [{\citenamefont {Khanian}\ and\ \citenamefont
  {Winter}(2020)}]{ZK_mixed_state_ISIT_2020}%
  \BibitemOpen
  \bibfield  {author} {\bibinfo {author} {\bibfnamefont {Z.~B.}\ \bibnamefont
  {Khanian}}\ and\ \bibinfo {author} {\bibfnamefont {A.}~\bibnamefont
  {Winter}},\ }\bibfield  {title} {\enquote {\bibinfo {title} {General mixed
  state quantum data compression with and without entanglement assistance},}\
  }in\ \href@noop {} {\emph {\bibinfo {booktitle} {Proc. IEEE Int. Symp. Inf.
  Theory (ISIT)}}}\ (\bibinfo {address} {Los Angeles, CA, USA},\ \bibinfo
  {year} {2020})\ pp.\ \bibinfo {pages} {1852--1857}\BibitemShut {NoStop}%
\bibitem [{\citenamefont {B.~Khanian}\ and\ \citenamefont
  {Winter}(2022)}]{general_mixed_state_compression}%
  \BibitemOpen
  \bibfield  {author} {\bibinfo {author} {\bibfnamefont {Z.}~\bibnamefont
  {B.~Khanian}}\ and\ \bibinfo {author} {\bibfnamefont {A.}~\bibnamefont
  {Winter}},\ }\bibfield  {title} {\enquote {\bibinfo {title} {General
  mixed-state quantum data compression with and without entanglement
  assistance},}\ }\href@noop {} {\bibfield  {journal} {\bibinfo  {journal}
  {IEEE Trans. Inf. Theory}\ }\textbf {\bibinfo {volume} {68}},\ \bibinfo
  {pages} {3130--3138} (\bibinfo {year} {2022})}\BibitemShut {NoStop}%
\bibitem [{\citenamefont {Hayden}\ \emph {et~al.}(2004)\citenamefont {Hayden},
  \citenamefont {Jozsa}, \citenamefont {Petz},\ and\ \citenamefont
  {Winter}}]{Hayden2004}%
  \BibitemOpen
  \bibfield  {author} {\bibinfo {author} {\bibfnamefont {P.}~\bibnamefont
  {Hayden}}, \bibinfo {author} {\bibfnamefont {R.}~\bibnamefont {Jozsa}},
  \bibinfo {author} {\bibfnamefont {D.}~\bibnamefont {Petz}}, \ and\ \bibinfo
  {author} {\bibfnamefont {A.}~\bibnamefont {Winter}},\ }\bibfield  {title}
  {\enquote {\bibinfo {title} {Structure of states which satisfy strong
  subadditivity of quantum entropy with equality},}\ }\href@noop {} {\bibfield
  {journal} {\bibinfo  {journal} {Commun. Math. Phys.}\ }\textbf {\bibinfo
  {volume} {246}},\ \bibinfo {pages} {359--374} (\bibinfo {year}
  {2004})}\BibitemShut {NoStop}%
\bibitem [{\citenamefont {Winter}(1999)}]{Winter1999}%
  \BibitemOpen
  \bibfield  {author} {\bibinfo {author} {\bibfnamefont {A.}~\bibnamefont
  {Winter}},\ }\emph {\bibinfo {title} {Coding Theorems of Quantum Information
  Theory}},\ \href@noop {} {\bibinfo {type} {{PhD} thesis}},\ \bibinfo
  {school} {Universit\"at Bielefeld, Department of Mathematics}, \bibinfo
  {address} {Germany} (\bibinfo {year} {1999}),\ \bibinfo {note}
  {arXiv:quant-ph/9907077}\BibitemShut {NoStop}%
\bibitem [{\citenamefont {Leditzky}\ \emph {et~al.}(2016)\citenamefont
  {Leditzky}, \citenamefont {Wilde},\ and\ \citenamefont
  {Datta}}]{Leditzky2016}%
  \BibitemOpen
  \bibfield  {author} {\bibinfo {author} {\bibfnamefont {F.}~\bibnamefont
  {Leditzky}}, \bibinfo {author} {\bibfnamefont {M.}~\bibnamefont {Wilde}}, \
  and\ \bibinfo {author} {\bibfnamefont {N.}~\bibnamefont {Datta}},\ }\bibfield
   {title} {\enquote {\bibinfo {title} {Strong converse theorems using
  {R}\'enyi entropies},}\ }\href@noop {} {\bibfield  {journal} {\bibinfo
  {journal} {J. Math. Phys.}\ }\textbf {\bibinfo {volume} {57}},\ \bibinfo
  {pages} {082202} (\bibinfo {year} {2016})}\BibitemShut {NoStop}%
\bibitem [{\citenamefont {Morgan}\ and\ \citenamefont
  {Winter}(2014)}]{AW_pretty}%
  \BibitemOpen
  \bibfield  {author} {\bibinfo {author} {\bibfnamefont {C.}~\bibnamefont
  {Morgan}}\ and\ \bibinfo {author} {\bibfnamefont {A.}~\bibnamefont
  {Winter}},\ }\bibfield  {title} {\enquote {\bibinfo {title} {``pretty
  strong'' converse for the quantum capacity of degradable channels},}\
  }\href@noop {} {\bibfield  {journal} {\bibinfo  {journal} {IEEE Trans. Inf.
  Theory}\ }\textbf {\bibinfo {volume} {60}},\ \bibinfo {pages} {317--333}
  (\bibinfo {year} {2014})}\BibitemShut {NoStop}%
\bibitem [{\citenamefont {Fuchs}\ and\ \citenamefont
  {de~Graaf}(1999)}]{Fuchs1999}%
  \BibitemOpen
  \bibfield  {author} {\bibinfo {author} {\bibfnamefont {C.~A.}\ \bibnamefont
  {Fuchs}}\ and\ \bibinfo {author} {\bibfnamefont {J.~van}\ \bibnamefont
  {de~Graaf}},\ }\bibfield  {title} {\enquote {\bibinfo {title} {Cryptographic
  distinguishability measures for quantum-mechanical states},}\ }\href@noop {}
  {\bibfield  {journal} {\bibinfo  {journal} {IEEE Trans. Inf. Theory}\
  }\textbf {\bibinfo {volume} {45}},\ \bibinfo {pages} {1216--1227} (\bibinfo
  {year} {1999})}\BibitemShut {NoStop}%
\bibitem [{\citenamefont {Terhal}\ \emph {et~al.}(2002)\citenamefont {Terhal},
  \citenamefont {Horodecki}, \citenamefont {Leung},\ and\ \citenamefont
  {DiVincenzo}}]{E_p}%
  \BibitemOpen
  \bibfield  {author} {\bibinfo {author} {\bibfnamefont {B.~M.}\ \bibnamefont
  {Terhal}}, \bibinfo {author} {\bibfnamefont {M.}~\bibnamefont {Horodecki}},
  \bibinfo {author} {\bibfnamefont {D.~W.}\ \bibnamefont {Leung}}, \ and\
  \bibinfo {author} {\bibfnamefont {D.~P.}\ \bibnamefont {DiVincenzo}},\
  }\bibfield  {title} {\enquote {\bibinfo {title} {The entanglement of
  purification},}\ }\href@noop {} {\bibfield  {journal} {\bibinfo  {journal}
  {J. Math. Phys.}\ }\textbf {\bibinfo {volume} {43}},\ \bibinfo {pages}
  {4286--4298} (\bibinfo {year} {2002})}\BibitemShut {NoStop}%
\bibitem [{\citenamefont {Stinespring}(1955)}]{Stinespring1955}%
  \BibitemOpen
  \bibfield  {author} {\bibinfo {author} {\bibfnamefont {W.~F.}\ \bibnamefont
  {Stinespring}},\ }\bibfield  {title} {\enquote {\bibinfo {title} {Positive
  {F}unctions on {$C^*$}-{A}lgebras},}\ }\href@noop {} {\bibfield  {journal}
  {\bibinfo  {journal} {Proc. Amer. Math. Society}\ }\textbf {\bibinfo {volume}
  {6}},\ \bibinfo {pages} {211--216} (\bibinfo {year} {1955})}\BibitemShut
  {NoStop}%
\bibitem [{\citenamefont {Werner}(1989)}]{Werner_extendible}%
  \BibitemOpen
  \bibfield  {author} {\bibinfo {author} {\bibfnamefont {R.~F.}\ \bibnamefont
  {Werner}},\ }\bibfield  {title} {\enquote {\bibinfo {title} {An application
  of bell’s inequalities to a quantum state extension problem},}\ }\href@noop
  {} {\bibfield  {journal} {\bibinfo  {journal} {Letters in Mathematical
  Physics}\ }\textbf {\bibinfo {volume} {17}},\ \bibinfo {pages} {359--363}
  (\bibinfo {year} {1989})}\BibitemShut {NoStop}%
\bibitem [{\citenamefont {Myhr}\ and\ \citenamefont
  {L\"{u}tkenhaus}(2009)}]{Myhr2009}%
  \BibitemOpen
  \bibfield  {author} {\bibinfo {author} {\bibfnamefont {G.~O.}\ \bibnamefont
  {Myhr}}\ and\ \bibinfo {author} {\bibfnamefont {N.}~\bibnamefont
  {L\"{u}tkenhaus}},\ }\bibfield  {title} {\enquote {\bibinfo {title} {Spectrum
  conditions for symmetric extendible states},}\ }\href@noop {} {\bibfield
  {journal} {\bibinfo  {journal} {Phys. Rev. A}\ }\textbf {\bibinfo {volume}
  {79}},\ \bibinfo {pages} {062307} (\bibinfo {year} {2009})}\BibitemShut
  {NoStop}%
\bibitem [{\citenamefont {Myhr}(2010)}]{Myhr_PhD}%
  \BibitemOpen
  \bibfield  {author} {\bibinfo {author} {\bibfnamefont {G.~O.}\ \bibnamefont
  {Myhr}},\ }\emph {\bibinfo {title} {Symmetric extension of bipartite quantum
  states and its use in quantum key distribution with two-way
  postprocessing}},\ \href@noop {} {\bibinfo {type} {{PhD} thesis}},\ \bibinfo
  {school} {Universit\"{a}t Erlangen-N\"{u}rnberg}, \bibinfo {address}
  {Germany} (\bibinfo {year} {2010}),\ \bibinfo {note}
  {arXiv:quant-ph/1103.0766}\BibitemShut {NoStop}%
\bibitem [{\citenamefont {Kaur}\ \emph {et~al.}(2021)\citenamefont {Kaur},
  \citenamefont {Das}, \citenamefont {Wilde},\ and\ \citenamefont
  {Winter}}]{Kaur2018}%
  \BibitemOpen
  \bibfield  {author} {\bibinfo {author} {\bibfnamefont {E.}~\bibnamefont
  {Kaur}}, \bibinfo {author} {\bibfnamefont {S.}~\bibnamefont {Das}}, \bibinfo
  {author} {\bibfnamefont {M.~M.}\ \bibnamefont {Wilde}}, \ and\ \bibinfo
  {author} {\bibfnamefont {A.}~\bibnamefont {Winter}},\ }\bibfield  {title}
  {\enquote {\bibinfo {title} {Resource theory of unextendibility and
  nonasymptotic quantum capacity},}\ }\href@noop {} {\bibfield  {journal}
  {\bibinfo  {journal} {Phys. Rev. A}\ }\textbf {\bibinfo {volume} {104}},\
  \bibinfo {pages} {022401} (\bibinfo {year} {2021})}\BibitemShut {NoStop}%
\bibitem [{\citenamefont {Christandl}\ and\ \citenamefont
  {Winter}(2005)}]{Christandl_EoP}%
  \BibitemOpen
  \bibfield  {author} {\bibinfo {author} {\bibfnamefont {M.}~\bibnamefont
  {Christandl}}\ and\ \bibinfo {author} {\bibfnamefont {A.}~\bibnamefont
  {Winter}},\ }\bibfield  {title} {\enquote {\bibinfo {title} {Uncertainty,
  monogamy and locking of quantum correlations},}\ }\href@noop {} {\bibfield
  {journal} {\bibinfo  {journal} {Proc. IEEE Int. Symp. Inf. Theory (ISIT)}\
  }\textbf {\bibinfo {volume} {51}},\ \bibinfo {pages} {879--883} (\bibinfo
  {year} {2005})}\BibitemShut {NoStop}%
\bibitem [{\citenamefont {Frank}\ and\ \citenamefont {Lieb}(2013)}]{Lieb2013}%
  \BibitemOpen
  \bibfield  {author} {\bibinfo {author} {\bibfnamefont {R.~L.}\ \bibnamefont
  {Frank}}\ and\ \bibinfo {author} {\bibfnamefont {E.~H.}\ \bibnamefont
  {Lieb}},\ }\bibfield  {title} {\enquote {\bibinfo {title} {Monotonicity of a
  relative {R}\'enyi entropy},}\ }\href@noop {} {\bibfield  {journal} {\bibinfo
   {journal} {J. Math. Phys.}\ }\textbf {\bibinfo {volume} {54}},\ \bibinfo
  {pages} {122201} (\bibinfo {year} {2013})}\BibitemShut {NoStop}%
\bibitem [{\citenamefont {Tomamichel}(2016)}]{Tomamichel_book}%
  \BibitemOpen
  \bibfield  {author} {\bibinfo {author} {\bibfnamefont {M.}~\bibnamefont
  {Tomamichel}},\ }\href@noop {} {\emph {\bibinfo {title} {Quantum Information
  Processing with Finite Resources}}}\ (\bibinfo  {publisher} {Springer Cham},\
  \bibinfo {year} {2016})\BibitemShut {NoStop}%
\bibitem [{\citenamefont {Choi}(1975)}]{Choi1975}%
  \BibitemOpen
  \bibfield  {author} {\bibinfo {author} {\bibfnamefont {M.}~\bibnamefont
  {Choi}},\ }\bibfield  {title} {\enquote {\bibinfo {title} {Completely
  positive linear maps on complex matrices},}\ }\href@noop {} {\bibfield
  {journal} {\bibinfo  {journal} {Linear Algebra Appl}\ }\textbf {\bibinfo
  {volume} {10}},\ \bibinfo {pages} {285--290} (\bibinfo {year}
  {1975})}\BibitemShut {NoStop}%
\bibitem [{\citenamefont {B.~Khanian}(in preparation)}]{Renyi_E_p}%
  \BibitemOpen
  \bibfield  {author} {\bibinfo {author} {\bibfnamefont {Z.}~\bibnamefont
  {B.~Khanian}},\ }\bibfield  {title} {\enquote {\bibinfo {title} {The
  {R}\'enyi entanglement of purification},}\ }\href@noop {} {\  (\bibinfo
  {year} {in preparation})}\BibitemShut {NoStop}%
\bibitem [{\citenamefont {Rockafeller}(1970)}]{Rockafeller}%
  \BibitemOpen
  \bibfield  {author} {\bibinfo {author} {\bibfnamefont {R.~T.}\ \bibnamefont
  {Rockafeller}},\ }\href@noop {} {\emph {\bibinfo {title} {Convex Analysis}}}\
  (\bibinfo  {publisher} {Princeton University Press},\ \bibinfo {year}
  {1970})\BibitemShut {NoStop}%
\bibitem [{\citenamefont {Dam}\ and\ \citenamefont {Hayden}()}]{Dam2002}%
  \BibitemOpen
  \bibfield  {author} {\bibinfo {author} {\bibfnamefont {W.}~\bibnamefont
  {Dam}}\ and\ \bibinfo {author} {\bibfnamefont {P.}~\bibnamefont {Hayden}},\
  }\bibfield  {title} {\enquote {\bibinfo {title} {R\'enyi-entropic bounds on
  quantum communication},}\ }\href@noop {} {\bibfield  {journal} {\bibinfo
  {journal} {preprint (2002)}\ }}\bibinfo {note}
  {ArXiv[quant-ph]:0204093}\BibitemShut {NoStop}%
\bibitem [{\citenamefont {Barnum}\ \emph {et~al.}(1996)\citenamefont {Barnum},
  \citenamefont {Fuchs}, \citenamefont {Jozsa},\ and\ \citenamefont
  {Schumacher}}]{Barnum1996}%
  \BibitemOpen
  \bibfield  {author} {\bibinfo {author} {\bibfnamefont {H.}~\bibnamefont
  {Barnum}}, \bibinfo {author} {\bibfnamefont {C.~A.}\ \bibnamefont {Fuchs}},
  \bibinfo {author} {\bibfnamefont {R.}~\bibnamefont {Jozsa}}, \ and\ \bibinfo
  {author} {\bibfnamefont {B.}~\bibnamefont {Schumacher}},\ }\bibfield  {title}
  {\enquote {\bibinfo {title} {General fidelity limit for quantum channels},}\
  }\href@noop {} {\bibfield  {journal} {\bibinfo  {journal} {Phys. Rev. A}\
  }\textbf {\bibinfo {volume} {54}},\ \bibinfo {pages} {4707--4711} (\bibinfo
  {year} {1996})}\BibitemShut {NoStop}%
\bibitem [{\citenamefont {Horodecki}\ \emph {et~al.}(2007)\citenamefont
  {Horodecki}, \citenamefont {Oppenheim},\ and\ \citenamefont
  {Winter}}]{Horodecki2007}%
  \BibitemOpen
  \bibfield  {author} {\bibinfo {author} {\bibfnamefont {M.}~\bibnamefont
  {Horodecki}}, \bibinfo {author} {\bibfnamefont {J.}~\bibnamefont
  {Oppenheim}}, \ and\ \bibinfo {author} {\bibfnamefont {A.}~\bibnamefont
  {Winter}},\ }\bibfield  {title} {\enquote {\bibinfo {title} {Quantum state
  merging and negative information},}\ }\href@noop {} {\bibfield  {journal}
  {\bibinfo  {journal} {Commun. Math. Phys.}\ }\textbf {\bibinfo {volume}
  {269}},\ \bibinfo {pages} {107--136} (\bibinfo {year} {2007})}\BibitemShut
  {NoStop}%
\bibitem [{\citenamefont {Dupuis}\ \emph {et~al.}(2014)\citenamefont {Dupuis},
  \citenamefont {Berta}, \citenamefont {Wullschleger},\ and\ \citenamefont
  {Renner}}]{One_Shot_Decoupling2014}%
  \BibitemOpen
  \bibfield  {author} {\bibinfo {author} {\bibfnamefont {F.}~\bibnamefont
  {Dupuis}}, \bibinfo {author} {\bibfnamefont {M.}~\bibnamefont {Berta}},
  \bibinfo {author} {\bibfnamefont {J.}~\bibnamefont {Wullschleger}}, \ and\
  \bibinfo {author} {\bibfnamefont {R.}~\bibnamefont {Renner}},\ }\bibfield
  {title} {\enquote {\bibinfo {title} {One-shot decoupling},}\ }\href@noop {}
  {\bibfield  {journal} {\bibinfo  {journal} {Commun. Math. Phys.}\ }\textbf
  {\bibinfo {volume} {328}},\ \bibinfo {pages} {251–284} (\bibinfo {year}
  {2014})}\BibitemShut {NoStop}%
\bibitem [{\citenamefont {Vitanov}\ \emph {et~al.}(2013)\citenamefont
  {Vitanov}, \citenamefont {Dupuis}, \citenamefont {Tomamichel},\ and\
  \citenamefont {Renner}}]{smooth_chain_rules}%
  \BibitemOpen
  \bibfield  {author} {\bibinfo {author} {\bibfnamefont {A.}~\bibnamefont
  {Vitanov}}, \bibinfo {author} {\bibfnamefont {F.}~\bibnamefont {Dupuis}},
  \bibinfo {author} {\bibfnamefont {M.}~\bibnamefont {Tomamichel}}, \ and\
  \bibinfo {author} {\bibfnamefont {R.}~\bibnamefont {Renner}},\ }\bibfield
  {title} {\enquote {\bibinfo {title} {Chain rules for smooth min- and
  max-entropies},}\ }\href@noop {} {\bibfield  {journal} {\bibinfo  {journal}
  {IEEE Trans. Inf. Theory}\ }\textbf {\bibinfo {volume} {59}},\ \bibinfo
  {pages} {2603–2612} (\bibinfo {year} {2013})}\BibitemShut {NoStop}%
\bibitem [{\citenamefont {Tomamichel}(2011)}]{Tomamichel_PhD}%
  \BibitemOpen
  \bibfield  {author} {\bibinfo {author} {\bibfnamefont {M.}~\bibnamefont
  {Tomamichel}},\ }\emph {\bibinfo {title} {A framework for non-asymptotic
  quantum information theory}},\ \href@noop {} {\bibinfo {type} {{PhD}
  thesis}},\ \bibinfo  {school} {ETH Zurich, Dept. Phys.}, \bibinfo {address}
  {Switzerland} (\bibinfo {year} {2011}),\ \bibinfo {note}
  {arXiv:quant-ph/1203.2142}\BibitemShut {NoStop}%
\bibitem [{\citenamefont {Renner}(2005)}]{Renner_PhD}%
  \BibitemOpen
  \bibfield  {author} {\bibinfo {author} {\bibfnamefont {R.}~\bibnamefont
  {Renner}},\ }\emph {\bibinfo {title} {Security of quantum key
  distribution}},\ \href@noop {} {\bibinfo {type} {{PhD} thesis}},\ \bibinfo
  {school} {ETH Zurich, Dept. Phys.}, \bibinfo {address} {Switzerland}
  (\bibinfo {year} {2005}),\ \bibinfo {note}
  {arXiv:quant-ph/0512258}\BibitemShut {NoStop}%
\end{thebibliography}%
\end{document}